\documentclass[conference,letterpaper,10pt,twocolumn]{ieeeconf}
 \usepackage{amssymb}
 \usepackage{amsmath}
\usepackage{multicol}
\usepackage{graphicx}
 \usepackage{algorithmicx}
 \usepackage[section]{algorithm}
 \usepackage{algpascal}
 \usepackage{algc}
 \usepackage{algcompatible}
 \usepackage{algpseudocode}
 \let\oldReturn\Return
 \renewcommand{\Return}{\State\oldReturn}
 \newtheorem{definition}{Definition}
 \usepackage{xcolor}
 \usepackage{subcaption} %to have subfigures available

 \newtheorem{theorem}{Theorem}
 
 \newtheorem{lemma}{Lemma}
 
 \newtheorem{assumption}{Assumption}

 \newtheorem{cor}{Corollary}

 \newtheorem{prop}{Proposition}
 
\usepackage[utf8]{inputenc}
\usepackage{amsmath,bm}
\def\href#1#2{#1 #2}

\title{\LARGE\bf Submodular Input Selection for Synchronization in Kuramoto  Networks}
\author{Dinuka Sahabandu$^1$, Andrew Clark$^2$, Linda Bushnell$^1$, and Radha Poovendran$^1$
	\thanks{$^1$D. Sahabandu, L. Bushnell, and R. Poovendran are with the Department of Electrical and Computer Engineering, University of Washington, Seattle, WA 98195 USA. \texttt{\{sdinuka,lb2,rp3\}@uw.edu}.}
	\thanks{$^2$A. Clark is with the Department of Electrical and Computer Engineering, Worcester Polytechnic Institute, Worcester, MA 01609 USA. \texttt{aclark@wpi.edu}.}
}
\IEEEoverridecommandlockouts
\begin{document}
\maketitle	
%%%%%%%%% Abstract
\begin{abstract}
Synchronization is an essential property of engineered and natural networked dynamical systems. The Kuramoto model of nonlinear synchronization has been widely studied in applications including entrainment of clock cells in brain networks and power system stability. Synchronization of Kuramoto networks has been found to be challenging in the presence of signed couplings between oscillators and when the network includes oscillators with heterogeneous natural frequencies. In this paper, we study the problem of minimum-set control input selection for synchronizing signed Kuramoto networks. We first derive sufficient conditions for synchronization in homogeneous as well as heterogeneous Kuramoto networks using a passivity-based framework. We then develop a submodular algorithm for selecting a minimum set of control inputs for a given Kuramoto network. We evaluate our approach through a numerical study on multiple classes of graphs, including undirected, directed, and  cycle graphs.
%for the desired functionality
% Kuramoto model that captures the dynamics of coupled nonlinear oscillators is a widely used in studying synchronization. 
%with both homogeneous and heterogeneous natural frequencies
\end{abstract}

%%%%%%%%% BODY TEXT
\section{Introduction}\label{sec:intro}

%Ability to reach an agreement among set of subsystems in a complex networked dynamical system has been taken a special attention among researchers recently due to the ever-growing interconnectedness between man-made systems~\cite{zoican2018blockchain,cheng2020reaching}. The phenomenon of coupled subsystems reaching an agreement is often called as consensus or synchronization. The study of synchronization is based on the networked graph structure and the underlying network dynamics~\cite{saber2003consensus}. There are many real-world engineering applications that entail synchronization, including but not limited to: consensus and formation control in multi-agent autonomous vehicle systems \cite{oh2015survey}, clock synchronization in wireless sensor network \cite{schenato2007distributed} and consensus on social networks \cite{mossel2009reaching}. On the other hand, scientific community has been also using the study of synchronization to shed the light on understanding many natural events such as bird flocking~\cite{chazelle2014convergence}, swarm intelligence in animal groups \cite{katsikopoulos2010swarm} and crowds of people~\cite{low2000following}.
Synchronization of coupled nonlinear oscillators has been extensively studied in engineering, physics, and biology.  Applications in engineering include power system stability \cite{dorfler2013synchronization} and clock synchronization in wireless sensor networks \cite{schenato2007distributed}. In physics, synchronization has been used to study coupled pendulum clocks \cite{bennett2002huygens}. It has been employed in biology to understand the circadian rhythms in networks of bursting neurons \cite{herzog2007neurons} and to investigate the synchronous firing of cardiac pacemaker cells \cite{torre1976theory}. 
%Some of the aforementioned applications require entire system to be synchronized while others require only some subsystems to be synchronized \cite{strogatz2000kuramoto}. 

Kuramoto  dynamics \cite{kuramoto1975self} is a widely accepted model used to study synchronization in networked oscillators due to its ability to capture a variety of real-world synchronization phenomena \cite{acebron2005kuramoto}. Under the Kuramoto model each oscillator's dynamics consists of two components, its own natural frequency and positively- or negatively-weighted sinusoidal couplings with the neighboring oscillators. 

Synchronization in oscillator networks has been studied under two categories: networks with homogeneous natural frequencies (e.g., synchronization of coupled identical pendulum clocks \cite{bennett2002huygens}) and networks with heterogeneous natural frequencies (e.g., study of cardiac rhythms \cite{torre1976theory}). In Kuramoto networks with homogeneous natural frequencies, it is known that all oscillators will converge to the same phase angle (phase synchronization) from any initial condition when the network is undirected and contains only positive couplings \cite{jadbabaie2004stability}. Synchronization, however, may not be realized in homogeneous Kuramoto dynamics on weakly connected directed graphs or under signed couplings \cite{strogatz2000kuramoto}. Moreover, even with all positive couplings, heterogeneous Kuramoto networks may fail to achieve synchronization~\cite{jadbabaie2004stability}. In \cite{clark2017toward}, a subset of oscillators were chosen to drive a positively coupled Kuramoto network towards synchronization.
%a mechanism to drive a Kuramoto network towards synchronization is to select 
%What mechanisms are required to bring such networks to synchronization remains an open problem.

%But they are not necessarily in synchronization when the networks are either weakly connected directed graphs or composed with signed couplings \cite{dorfler2012synchronization,bronski2012fully}. In contrast, Kuramoto networks with heterogeneous natural  frequencies are not guaranteed to be in synchronization even when the network is undirected and couplings are positive \cite{jadbabaie2004stability}.

Achieving synchronization in any Kuramoto network requires answering the following questions. First, given any Kuramoto network, can we find a mechanism that drives it towards synchronization? Second, what are the conditions to drive a given Kuramoto network towards synchronization? %

%An analytical framework for selecting such control inputs and providing desired  conditions for synchronization in both homogeneous and heterogeneous networks remain an open problem.  
%A mechanism to guarantee synchronization in inherently desynchronized Kuramoto networks is to select subset of oscillators in the network to act as input nodes.  Control inputs of desired values increase ability of driving desynchronized Kuramoto network back to synchronized state. Selection of input nodes for synchronization required answering the following two questions. $(i)$~What are the conditions required for synchronization? $(ii)$~How to choose minimum number of input nodes that guarantee synchronization.  

%Note that although the conditions mentioned in question~$(i)$ are studied extensively in literature for positively coupled Kuramoto networks \cite{kuramoto1975self,jadbabaie2004stability,dorfler2012synchronization}, signed Kuramoto networks has received the attention only in recent past \cite{hong2011kuramoto,el2013synchronization,delabays2019kuramoto}.

In this paper we develop an analytical framework for selecting a minimum set of control inputs with constant and identical phases and natural frequencies to drive any given Kuramoto network towards synchronization.  To the best of our knowledge, this is the first paper to do so. We make the following contributions.   

%we address the problem of
% In this paper, to the best of our knowledge,  we present an analytical framework to drive any given Kuramoto network towards synchronization.
%homogeneous as well as heterogeneous Kuramoto networks with signed coupling towards synchronization. We make the following contributions.   
%aforementioned problems and provide an analytical framework for selecting control inputs that will enable us to achieve synchronization in homogeneous as well as heterogeneous natural frequency cases.  The contributions of this paper are the following:
\begin{itemize}
	\item We decompose the Kuramoto model into a negative feedback interconnection between a nonlinear subsystem describing the local node dynamics and a linear subsystem describing the inter-node coupling. 
	\item We derive sufficient conditions for synchronization in a broad class of networks, including homogeneous and heterogenous oscillators, undirected and directed graphs, and networks with signed coupling.
	\item We prove that selecting a minimum set of control inputs to satisfy the sufficient conditions is a submodular optimization problem, and develop a polynomial-time approximate algorithm with provable optimality bounds.
	\item We evaluate  our approach via a numerical study.
	%\item We study the problem of selecting minimum number of inputs for the synchronization of signed Kuramoto networks with both homogeneous and heterogeneous natural frequencies over a set of general graph structures.
	%\item We propose an unified passivity based framework to study necessary/sufficient conditions required for phase/frequency synchronization in signed homogeneous/heterogeneous Kuramoto networks.
	%\item Under non-identical natural frequencies, we provide conditions on initial phase angles such that phase angle differences related to positive and negative couplings at frequency synchronization are trapped within the regions $[-\pi/2, \pi/2]$ and $[\pi/2, 3\pi/2]$, respectively.
	%and relate the problem of maximizing the minimum eigen value of a matrix by removing a set of rows and columns of the matrix.
	%\item We derive sufficient conditions for  frequency synchronization in signed homogeneous Kuramoto networks.
	%\item We obtain necessary conditions for phase synchronization in signed homogeneous Kuramoto networks.
	%\item We derive sufficient conditions for frequency synchronization in signed heterogeneous Kuramoto networks.
	%\item We propose a submodular optimization algorithm to find the minimum-size set of control input nodes required for the synchronization in homogeneous/heterogeneous Kuramoto networks.
	%\item We investigate the performance of proposed control input selection algorithm using a simulation study.
\end{itemize}

The rest of the paper is organized as follows. Section~\ref{sec:Related work} presents the related work. Section~\ref{sec:pre} provides needed preliminaries. Section~\ref{sec:game} presents the problem of selecting control inputs. Section~\ref{sec:results} presents the sufficient conditions for synchronization. Section~\ref{sec:alg} presents a submodular algorithm for selecting a minimum set of control inputs. Section~\ref{sec:sim} provides the simulation results. Section~\ref{sec:end} presents conclusions.

\section{Related work}\label{sec:Related work}

%The classical Kuramoto model introduced in \cite{kuramoto1975self} paved the way for studying synchronization in coupled nonlinear oscillators by deriving threshold on coupling strengths that enables synchronization in a networked infinite number of oscillators with all-to-all sames positive couplings. 

Most of the work in the literature focuses on studying the necessary/sufficient conditions to achieve phase/frequency synchronization in homogeneous/heterogeneous Kuramoto networks with positive couplings \cite{jadbabaie2004stability,bronski2012fully,rogge2004stability}. 
In \cite{jadbabaie2004stability,bronski2012fully} order parameter and Lyapunov theorem based approach was used to derive conditions for synchronization. In \cite{rogge2004stability}, conditions for frequency synchronization in a ring of unidirectionally coupled oscillators were studied by extending Gershgorin’s theorem. 
A passivity-based decomposition was proposed in \cite{mesbahi2010graph} to study synchronization in positively coupled, homogeneous Kuramoto networks of undirected graphs. We generalize this approach for deriving synchronization conditions in signed heterogeneous Kuramoto networks under various network topologies.
%the passivity based framework considered in our work.
%Frequency synchronization under heterogeneous natural frequencies has been analyzed in \cite{rogge2004stability} and %\cite{klein2009controlled} for positively coupled ring and chain networks, respectively. 

Recently, synchronization in signed Kuramoto networks with homogeneous/heterogeneous natural frequencies have received attention \cite{hong2011kuramoto,el2013synchronization,delabays2019kuramoto}. In \cite{hong2011kuramoto}  synchronization of Kuramoto oscillators with all-to-all signed couplings were studied. In \cite{el2013synchronization}  conditions for frequency synchronization in signed directed networks were derived. In \cite{delabays2019kuramoto}  the conditions that ensures the synchronization in signed acyclic directed oriented heterogeneous oscillator networks and signed oriented cyclic homogeneous oscillator networks were analyzed. Yet, the analytical frameworks considered in these papers do not generalize to studying synchronization conditions in signed Kuramoto networks with general undirected and directed graphs. 

%Deriving necessary and/or sufficient conditions for frequency synchronization in signed networks under some graph structures such as undirected and directed oriented graphs with cycles remain open.   

 %We generalize these results by deriving conditions for frequency synchronization in signed networks under more general classes of graph structures such as undirected and directed oriented graphs. 
 
 The effect of a single pacemaker (leader) on synchronization of Kuramoto networks with positive coupling under both homogeneous and heterogeneous natural frequencies  are studied in \cite{li2015synchronizing,wang2012exponential}. In \cite{clark2017toward} a submodular optimization framework for selecting a set of control input nodes in order to achieve synchronization in positively coupled Kuramoto networks was analyzed. A recent work in \cite{bosso2019global} developed a control strategy for the problem of leader-follower frequency synchronization in positively coupled Kuramoto networks, by exploiting the adaptive control framework. However, these models do not address synchronization of signed Kuramoto networks. A novel passivity-based framework and a submodular input selection algorithm presented in this paper enables deriving sufficient conditions for synchronization and steering a given Kuramoto network towards synchronization.
 %In contrast, we study the problem of driving an unsynchronized Kuramoto network to synchronization by selecting minimum number of input nodes and provide a submodular algorithm for input selection.
%Moreover we provide conditions on initial phase angles such that phase angle differences related to positive and negative couplings at phase locked state are trapped within the regions $[-\pi/2, \pi/2]$ and $[\pi/2, 3\pi/2]$, respectively. 

\section{Notation and Preliminaries}\label{sec:pre}

In this section we first introduce the notation used in this paper. Then we provide the background on passivity and  stability results from literature that have been used to derive the theoretical results presented in this paper.
\vspace{-1mm}
\subsection{Notation}

Let $R$ be a $m \times m$ matrix. Then $[R]_{ij}$ denotes the entry in $R$ corresponding to $i^{\text{th}}$ row and $j^{\text{th}}$ column, where $i,j \in \{1, \ldots, m\}$. Let $R^{T}$ denotes the transpose of matrix $R$. The notation $R >0$ implies that $R$ is a positive definite matrix. Let $\lambda_{\min}(R)$ denote the minimum eigenvalue of matrix $R$. The vectors $\mathbf{0}$ and $\mathbf{1}$ represent the all zeros and all ones vectors with appropriate dimensions, respectively. Let $\mathbb{E}(.)$ denote the standard expectation operator. The notations $||.||_{2}$ and $||.||_{\infty}$ denote the two-norm and infinity-norm of vectors.  
%In what follows, $t \in \mathbb{R}_{+}$ denotes the time.

%Let $\mathbb{R}$, $\mathbb{R}_{+}$, and $\mathbb{R}^{m}$ denote the real scalars, positive real scalars and $m$ dimensional real vectors, respectively.

%Matrix $R$ is symmetric if $R = R^{T}$
%For a finite set $S$,  let the number of elements denoted by $|S|$. Also let $\text{diag}(S)$ denotes a diagonal matrix of size $p\times p$, where $p \geq |S|$ with diagonal entries of the matrix $\text{diag}(S)$ corresponding to the indices in $S$ is set to one (i.e., $[\text{diag}(S)]_{ii} = 1$ for all $i \in S$) and all the other entries of $\text{diag}(S)$  are set to zero. 
\vspace{-1mm}
\subsection{Passivity and Stability}
For a dynamical system $\varSigma$ with input $u(t)$ and output $y(t)$, passivity is defined as follows.
%\begin{definition}[\cite{brogliato2007dissipative}]
%	\label{def:passivity}
%	$\varSigma$ is passive if there is a constant $\beta~\geq~0$ such that
%	\begin{equation*}
%	\int_{0}^{t}{y^{T}(\tau)u(\tau) \ d\tau} \geq \beta.
%	\end{equation*}
%\end{definition}
%\vspace{1.5mm}

%A differential condition for passivity in a dynamical system is as follows.
\begin{definition}[Passivity, \cite{brogliato2007dissipative}, Corollary 2.3]
	\label{prop:passivity-differential}
	Suppose that there exists a continuously differentiable function $V \geq 0$ and a measurable function $d$	such that $\int_{0}^{t}{d(s) \ ds} \geq 0$ for all $t$. Then if 
	\begin{equation}\label{inq:Passive}
	\dot{V}(t) \leq y^{T}(t)u(t) - d(t)
	\end{equation}
	for all $t$ and all $u(t)$, the system $\varSigma$ is passive. 
	%A system $\varSigma$ is said to be strictly passive if the inequality~\eqref{inq:Passive} is strict. 
\end{definition}

Next we introduce the notion of $\delta$-Input Strictly Passive ($\delta$-ISP) system in the following definition.

\begin{definition}[$\delta$-Input Strictly Passive ($\delta$-ISP) system]
	\label{def:ISP}
	Assume there exists a continuously differentiable function $V(\cdot) \hspace{-2mm}\geq \hspace{-2mm}0$ and a measurable function $\beta(\cdot)$ such that $\int_{0}^{t}{\beta(s)ds} \geq 0$ for all $t \geq 0$. Then a system $\varSigma$ is said to be $\delta$-Input Strictly Passive ($\delta$-ISP) if there exists a $\delta > 0$ such that 
	\begin{equation}\label{eq:ISP}
	\dot{V}(t) \leq y^{T}(t)u(t) - \delta u^{T}(t)u(t) - \beta(t)
	\end{equation}
	for all $t$ and all $u(t)$.
\end{definition}

The following proposition provides the conditions for $\mathcal{L}_{2}$ stability of interconnected subsystems.

\begin{prop}[\cite{brogliato2007dissipative}, Corollary 5.3]
	\label{prop:L2_stability_Kuramoto}
	%The following statements are true for a single channel negative feedback interconnection subsystems,
	 Consider two subsystems $H_{1}$ and $H_{2}$ in a single channel negative feedback interconnection. If $H_{1}$ is passive and $H_{2}$ is $\delta$-ISP:
	\begin{enumerate}
		\item Closed loop system is $\mathcal{L}_{2}$ finite gain stable.
		\item There exists a $\delta > 0$ such that $\mathcal{L}_{2}$ gain of the system is bounded above by $1/\delta $ .
	\end{enumerate}
\end{prop}

%\begin{prop}[\cite{brogliato2007dissipative}, Corollary 5.3]\label{prop:L2_stable}
%	A single channel negative feedback interconnection system in Figure~\ref{fig:Kuramoto-negative-feedback-nonidentical} is $\mathcal{L}_{2}$ finite gain stable if $H_1$ is passive and $H_2$ is Input Strictly Passive (ISP). 
%\end{prop}

The next result establishes a synchronization condition for a linear system with time-varying system matrix $A(t)$.

\begin{prop}[\cite{moreau2004stability}, Theorem~1]\label{prop: LTV consensus}
	Consider the system 
	\begin{equation}\label{eq:LTV}
	\dot{x}(t) = A(t)x(t),
	\end{equation}
	where $A(t)$ is a bounded and piecewise continuous function of $t$. For every $t$, $A(t)$ is Metzler with zero row sums. For any $\Delta > 0$ and any matrix $B$, let the $\Delta$~digraph to be defined as the digraph where an edge $(i, j)$ exists if $B_{ij} \geq 0$.
	
	 If there is an index $k \in \{1, \ldots, n\}$, a threshold value $\Delta > 0$ and an interval length $\hat{T} > 0$ such that for all $t \in \mathbb{R}$ the $\Delta$-digraph associated to $\int_{t}^{t + \hat{T}}A(\tau)d\tau$ has the property that all nodes may be reached from the node $k$, then the equilibrium set of synchronized states is uniformly exponentially stable. Solution of Eqn.~\eqref{eq:LTV} converge to a common value as $t \rightarrow \infty$.
\end{prop}

%It has been shown that the function $Q(E(S))$ in Proposition~\ref{prop:EigMaxLem} is increasing and submodular as a function of $E(S)$ \cite{clark2018maximizing}.
\section{Problem Formulation}\label{sec:game}
We consider a directed Kuramoto network of $n$ oscillators indexed $\{1, \ldots, n\}$ with couplings among the oscillators defined by the edge set $E:= \{1, \ldots, m\}$. An ordered pair $(i,j) \in E$ denotes an edge from oscillator $i$ to oscillator $j$ and implies that the oscillator $j$ is influenced by the oscillator $i$. Let the underlying graph structure of the Kurmaoto network is denoted by $\mathcal{G} = (V, E)$, where the set $V$ denotes the set of nodes\footnote{In this paper we use the words nodes and oscillators interchangeably.} that represent the Kuramoto oscillators. Define a vector of length $n$ by $\omega = \begin{bmatrix} \omega_{i}\end{bmatrix}^{n}_{i = 1}$ where each $\omega_{i} \in \mathbb{R}$ represents the natural frequency associated with the $i^{\text{th}}$ oscillator. Let $N_{\text{in}}(i):= \{j: (j,i) \in E\}$ denote the set of oscillators that have an incoming edge to oscillator $i$.

Kuramoto dynamics of the $i^{\text{th}}$ oscillator is given by
\begin{equation}
	%\label{eq:kuramoto}
	\dot{\theta}_{i}(t) = \omega_{i} -\sum_{j \in N_{\text{in}}(i)}{{K_{ji}}\sin{(\theta_{i}(t)-\theta_{j}(t))}}	
\end{equation}
where the coupling coefficients (i.e., edge weights) $K_{ji}$ are nonzero real numbers that characterize the influence of oscillator $j$ on oscillator $i$'s dynamics. We define the $n \times m$ incidence matrix $D$ by
\begin{displaymath}
	D_{ie} = \left\{
	\begin{array}{ll}
		1, & e=(j,i) \ \mbox{for some } j \\
		-1, & e=(i,j) \ \mbox{for some } j \\
		0, & \mbox{else}	
	\end{array}
	\right.
\end{displaymath}
We define a matrix $\hat{D}$ by 
\begin{displaymath}
	\hat{D}_{ie} = \left\{
	\begin{array}{ll}
		1, & e = (j,i) \ \mbox{for some } j \\
		0, & \mbox{else}	
	\end{array}
	\right.
\end{displaymath}
Next, we define $m \times m$ diagonal matrix $K$ by $K_{ee} = K_{ij}$ where $e = (i,j)$. Under these definitions, the dynamics of the network can be written in vector form as 
\begin{equation}
	%\label{eq:Kuramoto-vector}
	\dot{\theta}(t) = \omega -\hat{D}K\sin{(D^{T}\theta(t))}
\end{equation}
where $\sin{\theta} = (\sin{\theta_{i}}: i=1,\ldots,n)$. 

We assume that the network operates with a set of inputs, $S \subseteq \{1,\ldots,n\}$. Initial phase angles and natural frequencies associated with input nodes are set to zero (i.e, for all $i \in S$, $\theta_{i}(0) = 0$ and $\omega_{i} = 0$). Therefore, input nodes maintain a constant phase of zero regardless of the inputs of their neighbors. In such networks, the non-input node dynamics can be written as 
\begin{equation}\label{eq:Kuramoto with Inputs}
	\dot{\theta}_{i}(t) = \omega_{i} -\hspace{-3mm}\sum_{j \in N(i) \cap S}{{ \hspace{-2.5mm}K_{ji}}\sin{\theta_{i}(t)}} - \hspace{-2mm}\sum_{j \in N(i) \setminus S}{\hspace{-2.5mm}{K_{ji}}\sin{(\theta_{i}(t)-\theta_{j}(t))}},	
\end{equation}
while $\theta_{i}(t) \equiv 0$ for all $i \in S$. 

Let $E(S)$ denote the set of edges that are incoming to input nodes. In networks with inputs, we define the $(n-|S|) \times (m-|E(S)|)$ matrix $D(S)$   as
\begin{displaymath}
	D(S)_{ie} = \left\{
	\begin{array}{ll}	
		1, & i \notin S, e = (j,i) \ \mbox{ for some } j \\
		-1, & i \notin S, e = (i,j) \ \mbox{ for some } j \notin S \\
		0, & \mbox{else} 
	\end{array}
	\right.
\end{displaymath}
The effect of $S$ on $D$ is to remove all rows of $D$ related to $S$ and all columns of $D$ representing incoming edges to $S$. The $(n-|S|) \times (m-|E(S)|)$ matrix $\hat{D}(S)$ is defined as 
\begin{displaymath}
	\hat{D}(S)_{ie} = \left\{
	\begin{array}{ll}
		1, & i \notin S, e = (j,i) \mbox{ for some } j \\
		0, & \mbox{else}	
	\end{array}
	\right.
\end{displaymath}
Let $(m-|E(S)|)\times (m-|E(S)|)$ diagonal matrix $K(S)$ defined by $K_{ee} = K_{ji}$ where $e = (j,i)$ and $i \neq S$. Furthermore, with abuse of notation we let $\omega(S) = \begin{bmatrix} \omega_{i}\end{bmatrix}_{i \notin S}$.  Then Kuramoto model can be written as 
\begin{equation}
	\label{eq:Kuramoto-with-inputs}
	\dot{\theta}(t) = \omega(S) -\hat{D}(S)K(S)\sin{(D(S)^{T}\theta(t))}.	
\end{equation}

%Then, we multiply both sides of (\ref{eq:Kuramoto-with-inputs}) by $D(S)^{T}$ to obtain 
%\begin{equation}
%\label{eq:Kuramoto-with-inputs}
%D(S)^{T}\dot{\theta}(t) \hspace{-0.75mm} = \hspace{-0.75mm}D(S)^{T}\omega(S) -D(S)^{T}\hat{D}(S)K\sin{(D(S)^{T}\theta(t))}.	
%\end{equation}
Setting $z(t) = D(S)^{T}\theta(t)$, we arrive at the following equivalent model that gives the dynamics in terms of the inter oscillator phase angle differences, $z(t)$.
\begin{equation}
\label{eq:Kuramoto-with-inputs-equiv}
\dot{z}(t) = D(S)^{T}\omega(S) -D(S)^{T}\hat{D}(S)K(S)\sin{z(t)}.
\end{equation}

Next we define the notion of phase and frequency synchronization in Kuramoto networks. 
\begin{definition}[Phase synchronization]
	Kuramoto network is said to achieve phase synchronization if and only if $\lim_{t \rightarrow \infty}{\theta(t)} \rightarrow c_{p} \mathbf{1}$, where $c_{p} \in \mathbb{R}$. %and $\mathbf{1}$ represents a $m$-dimensional all ones vector. 
	%Using $z(t)$ phase synchronization can be interpreted as $\lim_{t \rightarrow \infty}{z(t)} \rightarrow \bar{c}_{p} $, where $\bar{c}_{p} \in \mathbb{R}$.
\end{definition}

\begin{definition}[Frequency synchronization]
	Kuramoto network is said to achieve frequency synchronization if and only if $\lim_{t \rightarrow \infty}{\dot{\theta}(t)} = c_{f} \mathbf{1}$, where $c_{f} \in \mathbb{R}$ or equivalently, $\lim_{t \rightarrow \infty}{\theta(t)} = \bar{c}_{f}$, where $\bar{c}_{f} \in \mathbb{R}^{m}$.
	%Using $z(t)$ this can be interpreted as $\lim_{t \rightarrow \infty}{z(t)} = \bar{c}_{f}$, where $\bar{c}_{f} \in \mathbb{R}^{m}$.
\end{definition}

 %\bar{c}_{p} =
Moreover, since the input nodes have fixed states of $0$, $c_{p} = c_{f} =  0$. In what follows, we use the word synchronization and frequency synchronization interchangeably. 

\noindent{\bf Problems Studied:} Selecting a minimum set of control inputs set to ensure (i)~frequency synchronization in Kuramoto networks with homogeneous natural frequencies~(Section~\ref{subsec:Freq_homo}), (ii)~phase synchronization in Kuramoto networks with homogeneous natural frequencies~(Section~\ref{subsec:phase_homo}) and (iii)~frequency synchronization in Kuramoto networks with heterogeneous natural frequencies~(Section~\ref{subsec:Freq_hetro}). 

%The problem analyzed in this paper is stated as follows: \emph{How to select a minimum-size set $S$ such that the dynamics $\theta(t)$ converges to phase/frequency synchronization?} We study this problem under two categories: (i) Kuramoto networks with homogeneous natural frequencies (i.e., $\omega_{i} = \omega_{j} $) and (ii) Kuramoto networks with heterogeneous natural frequencies (i.e., $\omega_{i} \neq\omega_{j} $).

\section{Analytical framework for selecting inputs}\label{sec:results}

Below, we present sufficient conditions for synchronization in the settings (i)-(iii) defined above.
%In this section we derive the passivity based sufficient conditions required for synchronization and provide an analytical framework for selecting a minimum-size set of control inputs to drive Kuramoto network to synchronization. We organize our results into three categories: Frequency synchronization with homogeneous natural frequencies, Phase synchronization with homogeneous natural frequencies and Frequency synchronization with heterogeneous  natural frequencies.

\subsection{Frequency synchronization: Homogeneous case}\label{subsec:Freq_homo}
 
Under this case, we first derive sufficient conditions for frequency synchronization and present a framework for selecting set of control inputs to drive Kuramoto network with homogeneous natural frequencies (i.e., $\omega_{i} = \omega_{j} $ for all $i,j$) to synchronization. First note that, by switching to a rotating frame, it can be shown that we can set $\omega_{i} \equiv 0$ for all $i \in \{1, \ldots, n\}$, without loss of generality \cite{strogatz2000kuramoto}.
 
 Let $M(S) \hspace{-0.5mm}\triangleq \hspace{-0.5mm}D(S)^{T}\hat{D}(S)K(S)$ and $R(S) \hspace{-0.5mm}\triangleq \hspace{-0.5mm}\frac{M(S) + M(S)^{T}}{2}$ (Matrices $M$ and $R$ are defined similarly by omitting the restriction to the set $S$). In order to derive sufficient conditions for synchronization, we represent the homogeneous Kuramoto network by a negative feedback interconnection between two subsystems as shown in Figure~\ref{fig:Kuramoto-negative-feedback} (with noting that $r(t) = 0$ when the natural frequencies, $\omega_{i} \equiv 0$ for all $i$). The advantage of this approach is that $H_{1}$ is nonlinear but does not depend on the input set $S$, while the bottom block depends on $S$ but is linear. 
 
 %we first analyze the stability of the system in Eqn.~\eqref{eq:Kuramoto-with-inputs-equiv} with $\omega_{i} \equiv 0$ for all $i$. In this respect, 
 
 %In order to gain additional insight, we represent the Kuramoto network with homogeneous natural frequencies (i.e., system  as a negative feedback interconnection between two subsystems, 

%\vspace{-1mm}
\begin{figure}[!ht]
	\centering
	\includegraphics[width=3.35in]{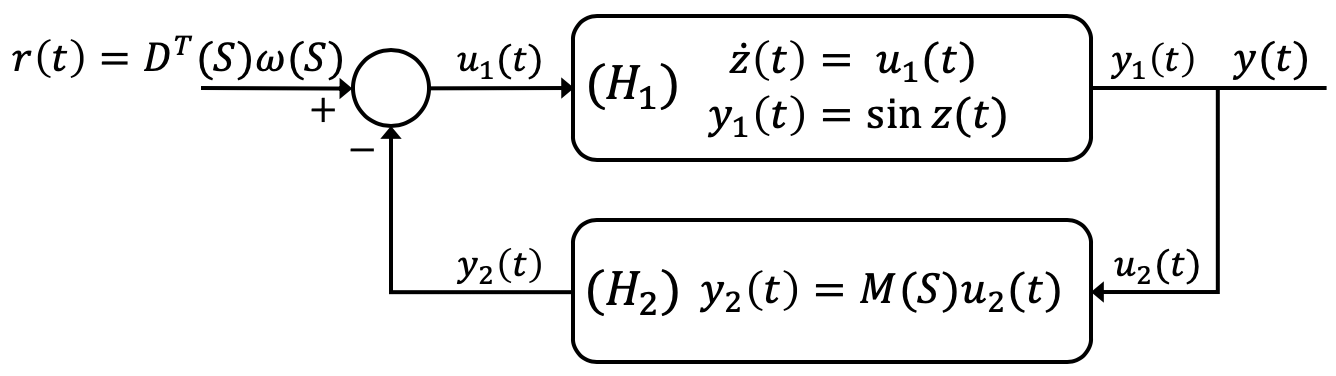}
	\caption{Decomposition of (\ref{eq:Kuramoto-with-inputs-equiv}) as a single channel negative feedback interconnection with a constant input $D^{T}(S)\omega(S)$.}%~=~D(S)^{T}\omega(S)$ $n$-dimensional .}
	\label{fig:Kuramoto-negative-feedback}	
\end{figure}
%\vspace{-1mm}
 %in Figure \ref{fig:Kuramoto-negative-feedback} via passivity. 
 %Then from Proposition~\ref{theorem:negative-feedback-passive} we need to identify sufficient conditions for passivity of $H_{1}$ and strict passivity of $H_{2}$ to establish the stability of homogeneous Kuramoto networks. 
%The following remark provide the motivation behind using this passivity approach.

%We first note that the system $H_{1}$ in Figure~\ref{fig:Kuramoto-negative-feedback}  is passive by Corollary~\ref{cor:top-block-passive}. Next 
%From Proposition~\ref{theorem:passivity-memoryless} we observe that $H_{2}$ is strictly passive if and only if $> 0$. The following lemma proves the passivity of $H_{1}$.

%\begin{lemma}\label{lem:H_1 passive}
%$H_{1}$ is passive.
%\newline
%\begin{proof}
%	\hspace{-3mm}
%	The result follows from Definition~\ref{prop:passivity-differential}.
%\end{proof}
%\end{lemma}

In the following we present a sufficient condition for frequency synchronization.  
%and using the fact that $R(S)$ is positive definite.
%Then, in order to show the stability of the Kuramoto network, it therefore suffices to show strict passivity of $H_{2}$ by Proposition~\ref{theorem:negative-feedback-passive}. 

\begin{theorem}\label{thm:suf_homo}
	A homogeneous Kuramoto network achieves frequency synchronization if $R(S) > 0$.
	\newline
	\begin{proof} 
		Consider the Lyapunov storage function $V(z) = \sum_{e \in E}{(1-\cos{z_{e}(t)})}.$ This storage function satisfies 
		\begin{equation}\label{eq:storage}
		\dot{V}(z) = \sum_{e \in E}{[\sin{z_{e}(t)}\dot{z}_{e}(t)]} = u_{1}(t)^{T}y_{1}(t). 
		\end{equation}
		Note that $y_{1}(t) = u_{2}(t)$ and $u_{1}(t) = -y_{2}(t)$. Therefore, 
		 %Therefore by Definition~\ref{def:passivity}, $H_{1}$ is passive.
		 \begin{equation*}
		 \begin{split}
		 \dot{V}(z) =& -u_{2}(t)^{T}y_{2}(t) = -u_{2}(t)^{T}M(S)u_{2}(t) \\
		                 =& -u_{2}(t)^{T}\frac{M(S) + M(S)^{T}}{2}u_{2}(t) \\
		                 =& -u_{2}(t)^{T}{R(S)}u_{2}(t) \leq 0
		 \end{split}
		 \end{equation*}
		 since $R(S) > 0$. Hence, by LaSalle's Invariance Principle, the state $z(t)$ converges to the set $\{z : \dot{V}(z) = 0\}$. Since $\dot{V}(z)$ is zero if and only if $u_{2}(t) = 0$ as $R(S)>0$, this set is equivalent to $\{z : \sin{z} = 0\}$, which is the discrete set of points satisfying $[z]_{e} = k_{e}\pi$ for some set of integers $\{k_{e} : e \in E\}$. Since the set of zeros of $\dot{V}(z)$ is discrete, the state trajectory $z(t)$ converges to a single fixed point, and thus $\lim_{t \rightarrow \infty}{\dot{z}(t)} = 0$, implying synchronization.
		\end{proof}
\end{theorem}

In order to develop a mechanism to select control inputs, we re-index the nodes and edges in $\mathcal{G}$ as follows. We assume that the set of non-input nodes are indexed in $\{1,\ldots,n-|S|\}$, and the set of input nodes is are indexed in $\{n-|S|+1,\ldots,n\}$. Let $\overline{S}$ denote the set of non-input nodes. For two distinct sets of nodes, $S_{1}$ and $S_{2}$, let $E(S_{1},S_{2})$ denote the set of edges from the nodes in $S_{1}$ to the nodes in $S_{2}$. We then partition the edge set as $E = E(\overline{S}, \overline{S}) \cup E(S, \overline{S}) \cup E(\overline{S},S) \cup E(S,S)$ and assume that the edges are indexed in this order. Furthermore, for a matrix $Q$ associated with graph $\mathcal{G}$ and set of nodes $S_3$, $S_4$ and $S_5$, let the matrix notation $Q^{S_3}_{S_{4}, S_{5}}$ denote the rows and columns of $Q$ restricted to the node set $S_3$ and edges from the nodes in $S_{4}$ to the nodes in $S_{5}$.

Using the new indices, we rewrite $M$ and $M(S)$ as
\begin{equation*}
\label{eq:M0}
M = \left(
\begin{array}{cccc}
(D_{\overline{S},\overline{S}}^{\overline{S}})^{T}\hat{D}_{\overline{S},\overline{S}}^{\overline{S}}K_{\overline{S},\overline{S}}^{\overline{S}} & (D_{\overline{S},\overline{S}}^{\overline{S}})^{T}\hat{D}_{S,\overline{S}}^{\overline{S}} K_{S,\overline{S}}^{\overline{S}}  \\
(D_{S,\overline{S}}^{\overline{S}})^{T}\hat{D}_{\overline{S},\overline{S}}^{\overline{S}} K_{\overline{S},\overline{S}}^{\overline{S}}& (D_{S,\overline{S}}^{\overline{S}})^{T}\hat{D}_{S,\overline{S}}^{\overline{S}}K_{S,\overline{S}}^{\overline{S}}  \\
(D_{\overline{S},S}^{\overline{S}})^{T}\hat{D}_{\overline{S},\overline{S}}^{\overline{S}} K_{\overline{S},\overline{S}}^{\overline{S}}& (D_{\overline{S},S}^{\overline{S}})^{T}\hat{D}_{S,\overline{S}}^{\overline{S}}K_{S,\overline{S}}^{\overline{S}} \\
0 & 0 
\end{array}
\right.
\end{equation*}
\hspace{3mm}
\begin{equation}
\label{eq:M}
\left.
\begin{array}{cccc}
0 & 0 \\
(D_{S,\overline{S}}^{S})^{T}\hat{D}_{\overline{S},S}^{S}K_{\overline{S},S}^{S} & (D_{S,\overline{S}}^{S})^{T}\hat{D}_{S,S}^{S}K_{S,S}^{S}  \\
(D_{\overline{S},S}^{S})^{T}\hat{D}_{\overline{S},S}^{S}K_{\overline{S},S}^{S}  & (D_{\overline{S},S}^{S})^{T}\hat{D}_{S,S}^{S}K_{S,S}^{S}  \\
(D_{S,S}^{S})^{T}\hat{D}_{\overline{S},S}^{S}K_{\overline{S},S}^{S} & (D_{S,S})^{T}\hat{D}_{S,S}^{S}	K_{S,S}^{S} 
\end{array}
\right)
\end{equation}

\begin{equation}
\label{eq:MS}
M(S) \hspace{-0.9mm}= \hspace{-0.9mm}\left(
\begin{array}{cc}
\hspace{-2mm}(D_{\overline{S},\overline{S}}^{\overline{S}})^{T}\hat{D}_{\overline{S},\overline{S}}^{\overline{S}} K_{\overline{S},\overline{S}}^{\overline{S}}& (D_{\overline{S},\overline{S}}^{\overline{S}})^{T}\hat{D}_{S,\overline{S}}^{\overline{S}}K_{S,\overline{S}}^{\overline{S}} \hspace{-2mm} \\
\hspace{-2mm}(D_{S,\overline{S}}^{\overline{S}})^{T}\hat{D}_{\overline{S},\overline{S}}^{\overline{S}} K_{\overline{S},\overline{S}}^{\overline{S}}& (D_{S,\overline{S}}^{\overline{S}})^{T}\hat{D}_{S,\overline{S}}	K_{S,\overline{S}}^{\overline{S}} \hspace{-2mm}
\end{array}
\right)
\end{equation}
Theorem~\ref{thm:input selection} gives a sufficient condition for synchronization via selecting set of control input nodes.

\begin{theorem}~\label{thm:input selection}
	Removing a subset of rows and columns from matrix $R$  such that $R(S) > 0$ suffices to drive a homogeneous Kuramoto network to frequency synchronization by selecting a set of control input nodes $S$.
	\newline
	\begin{proof}
		The matrix $M(S)$ in Eqn.~(\ref{eq:MS}) is equal to the $2 \times 2$ top-left block submatrix of $M$ in Eq. (\ref{eq:M}). Hence $M(S)$ can be obtained from $M$ by removing a subset of rows and columns. As consequence, the matrix $R(S)$ can be obtained from the matrix $R$ by removing a subset of rows and columns. Continuing this removal process until $R(S)>0$ drives the homogeneous Kuramoto networks towards the synchronization by Theorem~\ref{thm:suf_homo}.
		% we observe that it suffices to have $R(S)>0$ for 
		%Therefore, we attain synchronization in Kuramoto network with homogeneous natural frequencies by removing a subset of rows and columns from matrix $R$ such that $R(S)$ is positive definite.
		\end{proof}
\end{theorem}

\subsection{Phase synchronization: Homogeneous case}\label{subsec:phase_homo}
%\begin{theorem}[A Lyapunov exponential stability theorem]
%	Suppose there is a function $V$ and constant $\alpha > 0$ such that $V > 0$ and $\dot{V}(\theta) \leq -\alpha V(\theta)$. Then $||\theta(t) || \leq M \exp^{-\alpha t/2}||\theta(0) ||$
%\end{theorem}
The following theorem gives a necessary condition for phase synchronization in homogeneous Kuramoto networks.

\begin{theorem}
	Given $\theta_{i}(t)$ for all $i \notin S$ are initialized around the neighborhood of origin (i.e., $0$), homogeneous Kuramoto networks achieve phase synchronization if $R(S) > 0$.
	%synchronizes to a stable fixed point at the origin
	\newline
	\begin{proof}
		Consider the following linearized version of the equivalent system given in \eqref{eq:Kuramoto-with-inputs-equiv} (with $\omega \equiv  0$) at the origin.
		\begin{equation}\label{eq:linearized}
		\dot{\tilde{z}}(t) = \tilde{A}~\tilde{z}(t),
		\end{equation}
		where  $\tilde{A} = \frac{\partial \{-D^{T}\hat{D}K \sin z(t)\}}{\partial {z(t)}}\vert_{z(t) = 0} = -M(S).$
		%= -D(S)^{T}\hat{D}(S)K(S)\cos z(t)\vert_{z(t) = 0} 
		Then linear time-invariant system in Eqn.~\eqref{eq:linearized} is asymptotically stable if and only if all the real parts of the eigenvalues of $M(S)$ are positive. $R(S)>0$ ensures this condition. Since the system has a unique fixed point at $\tilde{z}(t) = 0$,  $z(t) \rightarrow 0$ when $t \rightarrow \infty$ and hence phase synchronization is achieved.
	\end{proof}
\end{theorem}
%(i.e., the system is phase synchronized)
Following the similar arguments as in Theorem~\ref{thm:input selection}, when the network is initialized in the neighborhood of the origin, removing a subset of rows and columns from matrix $R$  such that $R(S) > 0$ suffices to drive a homogeneous Kuramoto network to phase synchronization via selecting a set of control input nodes $S$. 
%However phase synchronization additionally requires the non-input oscillators in the network to be initialized in the neighborhood of $\mathbf{0}$. 

\subsection{Frequency synchronization: Heterogeneous case}\label{subsec:Freq_hetro}

In this section we provide the sufficient conditions for attaining frequency synchronization in heterogeneous Kuramoto networks (i.e., $\omega_{i} \neq\omega_{j} $ for at least one pair of  $(i, j)$ such that $i \neq j$) via selecting set of control input nodes. In what follows we write $D^{T}\omega$ instead of $D(S)^{T}\omega(S)$ for simplicity of notation. We discuss how to remove the dependency of $S$ in $D^{T}(S)\omega(S)$ after the proof of Theorem~\ref{theorem:signed-2}.

Our approach is based on analyzing the $\mathcal{L}_{2}$ stability of the system in Figure~\ref{fig:Kuramoto-negative-feedback} via passivity. Passivity of $H_{1}$ and $\delta$-ISP of $H_{2}$ together with the small gain theorem in Proposition~\ref{prop:L2_stability_Kuramoto}, gives the following bound on the $\mathcal{L}_{2}$ gain of the system.
%In order to gain additional insight, we illustrate the Kuramoto network with heterogeneous natural frequencies (i.e., system in Eqn.~\eqref{eq:Kuramoto-with-inputs-equiv}) as a single channel negative feedback interconnection system with two subsystems and an input channel with an $n$-dimensional constant input vector $r~=~D(S)^{T}\omega(S)$, as shown in Figure~\ref{fig:Kuramoto-negative-feedback}. 

%\begin{figure}[h!]
%	\centering
%	\includegraphics[width=3.5in]{Figures/non_identical_frequency.png}
%	\caption{Decomposition of (\ref{eq:Kuramoto-with-inputs-equiv}) as a single channel negative feedback interconnection.}
%	\label{fig:Kuramoto-negative-feedback-nonidentical}	
%\end{figure}

%Next in Theorem~\ref{thm:L2_stability_Kuramoto} we provide a sufficient condition for $\mathcal{L}_{2}$ stability of the Kuramoto model given in Eq.~\eqref{eq:Kuramoto-with-inputs-equiv-nonidentical}.
%Then the following result holds for the heterogeneous Kuramoto model given in Eq.~\eqref{eq:Kuramoto-with-inputs-equiv} by 

%of the single channel negative feedback interconnection illustrated in Figure~\ref{fig:Kuramoto-negative-feedback}, 
\begin{lemma}\label{cor:L2_stability_Kuramoto}
	
	Let $\delta = \lambda_{min}(R(S))$. Then for all $t$ and any $z(0)$, we have that  $$\left(||\sin{z(t)}||_{L_{2}}^{[0,t]}\right)^{2} \leq \frac{1}{\delta^{2}}||D^{T}\omega||_{2}^{2}t.$$
	\vspace{-2mm}
	\newline
	\begin{proof} 
		%To show passivity of $H_{1}$ from $u_{1}(t)$ to $y_{1}(t)$, choose the storage function $$V(z) = \sum_{e \in E}{(1-\cos{z_{e}(t)})}.$$ We then have $$\dot{V}(z) = \sum_{e \in E}{\sin{(z_{e}(t))}\dot{z}_{e}(t)} = u_{1}(t)^{T}\sin{z(t)} = u_{1}(t)y_{1}(t),$$ and hence $H_{1}$ is passive by Definition~\ref{def:passivity}.
		%We first show that $H_{1}$ is passive and $H_{2}$ is $\delta$-ISP, and then apply Proposition~\ref{prop:L2_stability_Kuramoto}. 
		Passivity of $H_{1}$ is confirmed via Eqn.~\eqref{eq:storage} by choosing the same storage function given in the proof of Theorem~\ref{thm:suf_homo}.
		
		Next, to show $\delta$-ISP of $H_{2}$, note that $$u_{2}(t)^{T}y_{2}(t) = u_{2}(t)^{T}R(S)u_{2}(t) \geq \lambda_{min}(R(S))u_{2}(t)^{T}u_{2}(t).$$ Thus by choosing $V(z) = 0$, we obtain $$\dot{V}(z) = 0 \leq u_{2}(t)^{T}y_{2}(t) - \lambda_{min}(R(S))u_{2}(t)^{T}u_{2}(t),$$ implying $\delta$-ISP with $\delta = {\lambda_{min}(R(S))}$.  
		
		By Proposition~\ref{prop:L2_stability_Kuramoto}, we have that for any input $r(t)$, $$\frac{||\sin{z(t)}||_{L_{2}}^{[0,T]}}{||r(t)||_{L_{2}}^{[0,T]}} \leq \frac{1}{\delta}.$$ Choosing $r(t) \equiv D^{T}\omega$ then yields $$||\sin{z(t)}||_{L_{2}}^{[0,T]} \leq \frac{1}{\delta}\left(\int_{0}^{T}{||D^{T}\omega||_{2}^{2} \ dt}\right)^{1/2}.$$ Squaring both sides and evaluating the integral gives the desired result.
		\end{proof}
	\end{lemma}

Let $z_{ji}(t) = \theta_{j}(t) - \theta_{i}(t)$ denote the inter oscillator phase angle differences between oscillator $j$ and $i$. In Lemma~\ref{lem:pre_suf}, we provide a preliminary sufficient condition required for the frequency synchronization in heterogeneous Kuramoto networks\footnote{Similar results have been appeared in literature for heterogeneous Kuramoto networks with positive couplings \cite{wang2012exponential,schmidt2012frequency}}. 

\begin{lemma}\label{lem:pre_suf}
	Heterogeneous Kuramoto networks achieve frequency synchronization if the phase angles of the oscillators are bounded for all $t$ such that $z_{ji}(t) \in (\pi/2, 3\pi/2)$ for all edges $(j,i)$ with $K_{ji} < 0$ and $z_{ji}(t) \in (-\pi/2, \pi/2)$ for all edges $(j,i)$ with $K_{ji} > 0$.
	\newline
	\begin{proof}
		By taking the time derivative of Eqn.~\eqref{eq:Kuramoto-with-inputs}, we have
		\begin{equation}
		\label{eq:time_derivative}
		\ddot{\theta}(t) = -\hat{D}(S)K_{c}(S)D(S)^{T}\dot{\theta}(t),
		\end{equation}
		where matrix $K_{c}(S)$ is a diagonal matrix with $[K_{c}(S)]_{ee} = K_{e}\cos(z_{e}(t))$, where $e = (j,i)$ for each $i,j \in \{1, \ldots, n\}$ and $i \notin S$. Note that matrix $\hat{D}(S)K_{c}(S)D(S)^{T}$ is the weighted Laplacian matrix of the graph $\mathcal{G}(S) = \{V\backslash S, E \backslash E(S)\}$ with time-varying weights $K_{e}\cos(z_{e}(t))$.
		
		Furthermore the following hold if each $z_{e}(t)$ satisfy the conditions given in theorem statement: (i) for all $e \in E$ such that $K_{e} < 0$ we have $cos(z_{e}(t)) < 0$ which gives $K_{e}cos(z_{e}(t)) > 0$. (ii) for all $e \in E$ such that $K_{e} > 0$, we have $\cos(z_{e}(t)) > 0$ which gives $K_{e}cos(z_{e}(t)) > 0$. %Therefore, all diagonal entries of $K_{c}(S)$ are positive for all $t$. 
		
		 When  all the weights, $[K_{c}(S)]_{ee}$, are positive for all $t$, the matrix $-D(S)^{T}\hat{D}(S)K_{c}(S)$ has nonnegative off diagonal entries (Metzler) with zero row sums. Then by Proposition~\ref{prop: LTV consensus}, $\ddot{\theta}(t)$ converge to a common value as $t \rightarrow \infty$.
	\end{proof}
\end{lemma}

Note that the bounds given in Lemma~\ref{lem:pre_suf} for $z(t)$ should satisfy for all $t$. Therefore, next we explore the initial conditions, $z_{ji}(0)$, that will achieve these bounds for all $t$.

In what follows we assume that the initial phase angles in Eqn.~\eqref{eq:Kuramoto-with-inputs-equiv} satisfies the following conditions.

\begin{assumption}\label{asmp:Phase_Angle_Init}
	The initial state $z(0)$ satisfies $z_{ij}(0) \in (\pi/2, 3\pi/2)$ for all edges $(i,j)$ with $K_{ji} < 0$ and $z_{ij}(0) \in (-\pi/2, \pi/2)$ for all edges $(i,j)$ with $K_{ji} > 0$. 
\end{assumption}
The following theorem establishes the sufficient condition for frequency synchronization in Kuramoto networks with heterogeneous natural frequencies.

The rest of this section is organized as follows. In Theorem~\ref{theorem:signed-2} we present sufficient conditions for frequency synchronization in heterogeneous Kuramoto networks.  Then we present set of additional results in Lemma~\ref{lemma:PI-1} and Theorem~\ref{theorem:signed} that are required for the proof of Theorem~\ref{theorem:signed-2}. Next, we present the proof of Theorem~\ref{theorem:signed-2} and discuss frequency synchronization in heterogeneous Kuramoto networks via selecting a minimum set of control inputs. Finally, we wrap-up this section by presenting set of graph structures that enable conditions given in Assumption~\ref{asmp:Phase_Angle_Init}.

%The following theorem presents the sufficient conditions required for synchronizing heterogeneous Kuramoto network by  
\begin{theorem}
	\label{theorem:signed-2}
	Suppose that the input set $S$ is chosen such that $\lambda_{\min}(R(S))$ is strictly bounded below by $||D^{T}\omega||_{2}$. If Assumption~\ref{asmp:Phase_Angle_Init} is satisfied then heterogeneous Kuramoto networks achieve frequency synchronization.
\end{theorem}

%In order to achieve frequency synchronization in Eqn.~\eqref{eq:Kuramoto-with-inputs-equiv} we need to first check whether $z(t)$ is upper bounded for all $t$. Thus in the following lemma we first establish the conditions for a nonnegative continuous function to be bounded above by a constant for all $t$.

In the following lemma we show that a nonnegative continuous function is upper bounded for all $t$ if its integral and initial conditions are upper bounded.
  
\begin{lemma}
\label{lemma:PI-1}
Let $f$ be a nonnegative continuous function. Suppose that, whenever $z(0)$ satisfies $f(z_{0}) \leq C$ for a constant $C$, $$\int_{0}^{t}{f(z(\tau)) \ d\tau} \leq tC$$ for all $t$. 
Then for any $z(0)$ satisfying $f(z(0)) \leq C$, we have $f(z(t)) \leq C$ for all $t$.	
\end{lemma}

\begin{proof}
 %\textcolor{blue}{Change from infinity norm to f(z) everywhere here.}
Suppose the result does not hold. Then there exists $z$ and $t$ such that $f(z(t)) > C$ for some time $t$ when $z(0) = z$. Since the trajectory of $z(t)$ is continuous, $f(z(t))$ is continuous as well. Define some notations as follows: 
%$t^{\ast} = \inf{\{t : f(z(t)) > C\}}$, $t^{\ast\ast} = \inf{\{t : f(z(t)) \leq C, \ t > t^{\ast}\}} $, $\overline{C} = \sup{\{f(z(t)) : t \in [t^{\ast}, t^{\ast\ast}]\}} $, $\eta =\overline{C}-C$, $t_{1} = \inf{\{t: f(z(t)) = C+\eta/2, t_{1} > t^{\ast}\}}$, and $t_{2} = \inf{\{t : f(z(t)) > C + \eta/2, t > t_{1}\}}$.
\begin{eqnarray*}
t^{\ast} &=& \inf{\{t : f(z(t)) > C\}} \\
t^{\ast\ast} &=& \inf{\{t : f(z(t)) \leq C, \ t > t^{\ast}\}} \\
\overline{C} &=& \sup{\{f(z(t)) : t \in [t^{\ast}, t^{\ast\ast}]\}} \\
\eta &=& \overline{C}-C \\
t_{1} &=& \inf{\{t: f(z(t)) = C+\eta/2, t_{1} > t^{\ast}\}} \\
t_{2} &=& \inf{\{t : f(z(t)) > C + \eta/2, t > t_{1}\}} 	
\end{eqnarray*}	
For any $\epsilon \hspace{-0mm}> 0$, let $t_{0}(\epsilon) = \sup{\{t : f(z(t)) \leq C-\epsilon, t< t^{\ast}\}}.$ Note that $(t^{\ast}-t_{0}(\epsilon))$ goes to zero as $\epsilon$ goes to zero. Then
	\begin{eqnarray*}
	\int_{t_{0}(\epsilon)}^{t^{\ast\ast}}{f(z(t)) \ d\tau} \hspace{-2mm}&=& \hspace{-2mm}\int_{t_{0}(\epsilon)}^{t^{\ast}}{f(z(t)) \ d\tau} + \int_{t^{\ast}}^{t_{1}}{f(z(t))\ d\tau} \\
	\hspace{-3mm}&& \hspace{-2mm}+ \int_{t_{1}}^{t_{2}}{f(z(t)) \ d\tau} + \int_{t_{2}}^{t^{\ast\ast}}{f(z(t)) \ d\tau} \\
	\hspace{-3mm}&>& \hspace{-2mm}(t^{\ast}-t_{0}(\epsilon))(C-\epsilon) + C(t_{1}-t^{\ast}) \\
	\hspace{-3mm}&& \hspace{-2mm}+ (C+\frac{\eta}{2})(t_{2}-t_{1}) + C(t^{\ast\ast}-t_{2}) \\
	\hspace{-3mm}&=&\hspace{-2mm} (t^{\ast\ast}-t_{0}(\epsilon))C + \frac{\eta}{2}(t_{2}-t_{1}) \\
	\hspace{-3mm}&&\hspace{-2mm}- \epsilon(t^{\ast}-t_{0}(\epsilon)) > C(t^{\ast\ast}-t_{0}(\epsilon)).
	\end{eqnarray*}
for $\epsilon$ sufficiently small. Let $z_{0} \hspace{-0.5mm}= \hspace{-0.7mm} z(t_{0}(\epsilon))$, so that $f(z(t))\hspace{-0.7mm}\leq\hspace{-0.5mm}C$. Set $z(0) = z_{0}$ and choose $t = t^{\ast\ast}-t_{0}(\epsilon)$. By time invariance, $\int_{0}^{t}{f(z(t))\ d\tau} > Ct$ is a contradiction.
\end{proof}

Theorem~\ref{theorem:signed}, shows that the function $\sin{z(t)}$ has lower and upper bounds, $-1 + \epsilon$ and $1 - \epsilon$ for $0 < \epsilon < 1$ when Assumption~\ref{asmp:Phase_Angle_Init} is satisfied and the smallest eigenvalue of $R(S)$ is strictly bounded below by $||D^{T}\omega||_{2}$ and $z(0)$.

\begin{theorem}
\label{theorem:signed}
Suppose that the input set $S$ is chosen such that the $\lambda_{\min}$ is strictly bounded below by $||D^{T}\omega||_{2}$. If Assumption~\ref{asmp:Phase_Angle_Init} is satisfied then $||\sin{z(t)}||_{\infty} \hspace{-0.5mm}<\hspace{-0.5mm}1$ for all $t$. 
\end{theorem}

\begin{proof}
Let $f(z) = ||\sin{z(t)}||_{2}^{2}$. Let $\epsilon$ satisfy $||D^{T}\omega||_{2}^{2}/\delta^{2} = 1-\epsilon$. Then from Lemma~\ref{cor:L2_stability_Kuramoto}, for all $t$, 
\begin{eqnarray*}
\int_{0}^{T}{||\sin{z(t)}||_{2}^{2} \ dt} \leq \frac{1}{\delta^{2}}||D^{T}\omega||_{2}^{2} T = (1-\epsilon)T	
\end{eqnarray*}
	Then by Lemma \ref{lemma:PI-1},  $\hspace{-0.5mm}||\sin{z(t)}||_{\infty}^{2} \hspace{-0.5mm}\leq\hspace{-0.5mm} ||\sin{z(t)}||_{2}^{2} \leq (1-\epsilon)$.
\end{proof}

In the following we provide the proof of Theorem~\ref{theorem:signed-2} using the results from Lemma~\ref{lem:pre_suf} and Theorem~\ref{theorem:signed}.

\textit{Proof of Theorem~\ref{theorem:signed-2}:}
%\begin{proof}
%By Proposition~\ref{prop: LTV consensus}, frequency synchronization is achieved if the Laplacian matrix has nonnegative weights on all edges. The weights are given by $$K_{ji}\cos{z_{ji}(t)}.$$ Hence if $z_{ji}(t) \in [\pi/2, 3\pi/2]$ for all $t$ when $K_{ji} < 0$ and $z_{ji}(t) \in [-\pi/2, \pi/2]$ for all $t$ when $K_{ji} > 0$, 	then the weights are positive for all time $t$.
Let $e = (j,i)$ for each $i,j \in \{1, \ldots, n\}$ and $i \notin S$. Then using contradiction we first prove that $z_{e}(t)$ satisfies the bounds given in Lemma~\ref{lem:pre_suf} if $z_{e}(0)$ satisfies the bounds in Assumption~\ref{asmp:Phase_Angle_Init}.

Suppose $z_{e}(t)$ does not satisfies the bounds in Lemma~\ref{lem:pre_suf} when $z_{e}(0)$ satisfies the bounds in Assumption~\ref{asmp:Phase_Angle_Init}. Suppose there exists $t$ such that $z_{e}(t) \notin [\pi/2, 3\pi/2]$ for some $e$ with $K_{e} < 0$. Since $z_{e}(0)$ is in this region and $z_{e}$ is continuous, there exists $t^{\ast} < t$ such that $z_{e}(t^{\ast}) \in \{\pi/2, 3\pi/2\}$. 

This, however, implies $\sin{z_{e}(t^{\ast})} \in \{-1,1\}$, contradicting Theorem~\ref{theorem:signed}. The contradiction in the case where there exists $t$ with $z_{e}(t) \notin [-\pi/2,\pi/2]$ for some $e$ with $K_{e} > 0$ is similar. Then the result follows from Lemma~\ref{lem:pre_suf}. $\null\hfill{\blacksquare}$
%\end{proof}

Following the similar arguments as in Theorem~\ref{thm:input selection}, when a heterogeneous Kuramoto network satisfies conditions in Theorem~\ref{theorem:signed-2}, removing a subset of rows and columns from matrix $R$  such that $\lambda_{\min}(R(S))$ is bounded below by $||D^{T}\omega||_{2}$ suffices to attain frequency synchronization via selecting a set of control input nodes $S$. We note that $||D^{T}\omega||_{2}$ depends on the set $S$. In order to remove this dependency, we can replace $||D^{T}\omega||_{2}$ with an uniform upper bound independent of $S$. Derivation of one such upper bound is given below.

\begin{theorem}\label{thm:bound}
Let $\bar{\delta}$ be defined as
\begin{equation}\label{eqn:upper_bound}
\bar{\delta}  = \sqrt{\sum_{(j,i) \in E} \max\{(\omega_{j} - \omega_{i})^{2}, \omega_{i}^{2}, \omega_{j}^{2}\}}.
\end{equation}
Then $\max_{S}{\{||D(S)^{T}\omega(S)||_{2}\}} \leq \overline{\delta}$.
\newline
\begin{proof}
	We have that 
		\begin{equation*}
		\begin{split}
		||D(S)^{T}\omega(S)||_{2}^{2} &= \sum_{(j,i) \in E(\overline{S},\overline{S})}(\omega_{j} - \omega_{i})^{2} + \sum_{(j,i) \in E(S,\overline{S})}\omega_{i}^{2} \\
			                                  &\leq \sum_{(j,i) \in E} \max\{(\omega_{j} - \omega_{i})^{2}, \omega_{i}^{2}, \omega_{j}^{2}\}.
		\end{split}
		\end{equation*}
		The results follows by taking the square root of the both sides of the above inequality.
\end{proof}
\end{theorem}

The set of conditions in Assumption~\ref{asmp:Phase_Angle_Init} is not feasible in general for all types of network graphs. For example, a graph with oscillators $i$ and $j$ with differently signed couplings between $(i,  j)$ and $(j,  i)$ does not satisfy the conditions in Assumption~\ref{asmp:Phase_Angle_Init}. Hence,  next we study set of network graph structures that satisfy the conditions given in  Assumption~\ref{asmp:Phase_Angle_Init}.

Let $E_p^{l}(i,j)$ and $E_n^{l}(i,j)$ denote the number of positive and negative edges in a path $l$  between the nodes $i$ and $j$ in $\mathcal{G}$ such that path length is strictly larger than one, respectively. Then the following results in Theorem~\ref{thm:Undirected_init}, Corollary~\ref{cor:Cycle init} and Corollary~\ref{cor:Tree_init}  present set of graphs that satisfy the conditions given in  Assumption~\ref{asmp:Phase_Angle_Init}. Proofs of Theorem~\ref{thm:Undirected_init}, Corollary~\ref{cor:Cycle init} and Corollary~\ref{cor:Tree_init} are given in Appendix.

\begin{theorem}\label{thm:Undirected_init}
	%$\mathcal{G_{\text{SSC}}}$ and $\mathcal{G_{\text{DOC}}}$ satisfy 
	Any undirected or directed oriented graph that fulfills one of the following conditions for each $(i,j) \in E$ satisfies Assumption~\ref{asmp:Phase_Angle_Init}.
	\begin{enumerate}
		\item If $K_{ij} > 0$, then for a path $l$ such that $(E_p^{l}(i,j) - E_n^{l}(i,j)) \geq 0$ requires $(E_p^{l}(i,j) - E_n^{l}(i,j))\mod4 \in \{0,1\}$ and for a path $l$ such that $(E_n^{l}(i,j) - E_p^{l}(i,j)) \geq 0$ requires $(E_n^{l}(i,j) - E_p^{l}(i,j))\mod4 \in \{1,3\}$ 
		\item If $K_{ij} < 0$, then for a path $l$ such that $(E_p^{l}(i,j) - E_n^{l}(i,j)) \geq 0$ requires $(E_p^{l}(i,j) - E_n^{l}(i,j))\mod4 \in \{2,3\}$ and for a path $l$ such that $(E_n^{l}(i,j) - E_p^{l}(i,j)) \geq 0$ requires $(E_n^{l}(i,j) - E_p^{l}(i,j))\mod4 \in \{1,2\}$
	\end{enumerate}
	%\begin{proof}
		
	%end{proof}
\end{theorem}

\begin{cor}\label{cor:Cycle init}
	Let $E_p$ and $E_n$ denote the number of positive and negative edges in a cycle graph. Then the underlying cycle graph satisfies Assumption~\ref{asmp:Phase_Angle_Init} if one the following conditions are met. 
	%(i)~$(E_p - E_n) \geq 1$ and $(E_p - E_n)\mod4 \in \{1,2\}$ or (ii)$(E_n - E_p) \geq  2$ and $(E_n - E_p)\mod4 \in \{2,3\}$
	\begin{enumerate}
		\item $(E_p - E_n) \geq 1$ and $(E_p - E_n)\mod4 \in \{1,2\}$
		\item $(E_n - E_p) \geq  2$ and $(E_n - E_p)\mod4 \in \{2,3\}$
	\end{enumerate}
\end{cor}

\begin{cor}\label{cor:Tree_init}
	Any Tree graph will satisfy Assumption~\ref{asmp:Phase_Angle_Init}.
	%\begin{proof}
		
		%\end{proof}
\end{cor}

\section{Submodular Algorithm for selecting a minimum set of control inputs}\label{sec:alg} 

In this section we provide an algorithm for selecting a minimum set of control inputs for phase/frequency synchronization in homogeneous Kuramoto networks and frequency synchronization in heterogeneous Kuramoto networks.

Using the results derived in Section~\ref{sec:results}, we formulate the minimum-set control input selection problem as 
\begin{equation}\label{prob:min_eig}
\begin{split}
&\min |S| \\
&\text{s.t.}~\lambda_{\min}(R(S)) = \delta
\end{split}
\end{equation}

Note that when $\delta = 0$ phase and frequency synchronization is achieved in homogeneous Kuramoto networks. 
On the contrary, when  $\delta = \bar{\delta}$ (in Eqn.~\eqref{eqn:upper_bound}) frequency synchronization is achieved in heterogeneous Kuramoto networks.
 %Analogously , let ${E}(S)$ denote the set of edges incoming to all the nodes $i \in S$. Define $\text{\textit{diag}}({E}(S))$ to be a $m\times m$ diagonal matrix with $i^{\text{th}}$ diagonal entry is $1$ if $e_{i} \in {E}(S)$ for $i = 1, \hdots, m$ and $0$ otherwise. 

%Now we present the following result from \cite{clark2018maximizing} which allows us to interpret the problem of selecting minimum set of input nodes to guarantee the phase synchronization of a Kuramoto network as a problem of maximizing the smallest eigenvalue of a symmetric matrix $R$ by removing set of rows and columns.
In what follows, we show that Problem~\ref{prob:min_eig} can be written as a submodular optimization problem.  The formal definition of submodularity of a set function $f$ is given below. 

\begin{definition}
	A set function $f: 2^{V} \rightarrow \mathbb{R}$ is submodular if and only if, for any $S \subseteq T \subseteq V$ and any $v \notin T$,
	$$f(S \cup \{v\}) - f(S) \geq f(T \cup \{v\}) - f(T),$$
	where $2^{V}$ denote the set of all subsets of $V$.
\end{definition}

The following proposition relates the problem of bounding the minimum eigenvalue of a symmetric matrix $R$ below by a constant $\delta>0$ after removing set of corresponding rows and columns in matrix $R$ to a function $Q(E(S))$ which is increasing and submodular in set $E(S)$.
\begin{prop}[\cite{clark2018maximizing}, Lemma 2]
	\label{prop:EigMaxLem}
	Let $E := \{1, \ldots, m\}$ denote the set of indices related to corresponding rows and columns in a $m \times m$ symmetric matrix $R$. Then for any subset, ${E}(S) \subseteq E$ and define a  diagonal matrix of size $m\times m$ by $\text{\textit{diag}}(E(S))$ with $[\text{\textit{diag}}(E(S))]_{ii} = 1$ for all $i \in E(S)$ and $[\text{\textit{diag}}(E(S))]_{ii} = 0$ for all $i \notin E(S)$. Then the following statements are equivalent:
	\begin{enumerate}
		\item  $\lambda_{\min}(R(E \backslash {E}(S))) > \delta$.
		\item There exists a constant $\alpha > 0$ such that $$\lambda_{\min}(R + \alpha \text{\textit{diag}}(E(S))) > \delta.$$
		\item If $\boldsymbol{w}$ is an $m$-dimensional Gaussian random vector with mean $0$ and covariance matrix $I$, then
		\begin{equation}\label{eq:Sub_Q}
		\hspace{-0.5mm} Q(E(S)) \hspace{-1mm}:= \mathbb{E}\big(\min \{w^{T}Rw+\alpha \hspace{-3mm}\sum_{i \in E(S)}\hspace{-2mm}w_{i}^{2}, \delta\}\big) \hspace{-1mm}= \hspace{-0.5mm}\delta.
		\end{equation}
	\end{enumerate}
\end{prop}

Using the results presented in Proposition~\ref{prop:EigMaxLem}, we rewrite the problem of minimum-set control input selection as follows.
\begin{equation}\label{prob_set_cover}
	\begin{split}
		&\min |S| \\
		&~\text{s.t.}~Q(E(S)) = \delta
	\end{split}
\end{equation}

Let  ${E}(i) := \{e~|~e=(j,i)~\text{and}~j \in N_{\text{in}}(i)\}$ denote the set of edges incoming to oscillator $i \in \{1, \hdots, n\}$ in the Kuramoto network. Problem~\eqref{prob:min_eig} leads to the following submodular minimum-set control input selection algorithm. Proposition~\ref{prop:sub_bounds} provides the optimality bounds of Algorithm~\ref{alg:input_selection}.
%for synchronization.
%\vspace{mm}
\begin{algorithm}[h]
	\vspace{1mm}
	\caption{Algorithm for selecting a minimum set of control inputs to induce frequency/phase synchronization.% in Kuramoto networks. 
		\label{alg:input_selection}}
	\begin{algorithmic}
		\State \textit {\bf Input: } Symmetric Matrix $R$, $$\delta=
		\begin{cases}
			0,~\text{if network is homogeneous }\\
			\bar{\delta} ~\text{in~Eqn}.~\eqref{eqn:upper_bound},~\text{if network is heterogeneous}
		\end{cases}$$
		\State \textit{\bf Output: } Input node set $S$
	\end{algorithmic}
	\begin{algorithmic}[1]
		\State $S \leftarrow\emptyset$
		\While {$Q(E(S)) < \delta$}
		\State $v^{*} \leftarrow \text{arg max}\{Q(E(S) \cup E(i)) : i \notin S\}$
		\State $S \leftarrow \{v^{*}\}$
		\State  Calculate $Q(E(S))$ using Eqn.~\eqref{eq:Sub_Q} with $R = R(S)$
		\EndWhile
		\Return $S$
	\end{algorithmic}
\end{algorithm}

%In Algorithm~\ref{alg:input_selection} Compute $\alpha$ such that $R(E \backslash E_S) + \alpha \text{\textit{diag}}(E_S) > 0$.

% \newpage
\begin{prop}\label{prop:sub_bounds}
	Let $T$ denote the total number of iterations Algorithm~\ref{alg:input_selection} takes to terminate. Then define  $S_T$ and $S^{*}$ to be the solution returned by the Algorithm~\ref{alg:input_selection} and optimal solution of the Problem given in \eqref{prob_set_cover}, respectively. The Algorithm~\ref{alg:input_selection} has the following optimality bound of $\log\frac{\delta - \lambda_{\min}(R)}{\delta - Q(E(S_{T-1}))}$, where $S_{T-1}$ denote the set of input nodes returned by algorithm in the iteration $T-1$. 
	%\begin{equation}
	%	|S_T| - |S^{*}| \leq |S^{*}| \log\frac{\lambda_{\min}(R)}{Q(E(S_{T-1}))}
%	\end{equation}
\newline
	\begin{proof}
		Notice that Algorithm~\ref{alg:input_selection} is a greedy algorithm that solves the submodular set covering problem in \eqref{prob_set_cover}. From \cite{wolsey1982analysis} we obtain the  following optimality bounds for greedy algorithm that considers the submodular function $Q$.
		\begin{equation}
		|S_T| - |S^{*}| \leq |S^{*}| \log\frac{Q(E(S_{T}))-Q(\emptyset)}{Q(E(S_{T}))-Q(E(S_{T-1}))} 
		\end{equation}
		The result follows by observing $Q(E(S_{T})) \hspace{-1mm}= \hspace{-1mm}\delta$ at the convergence of Algorithm~\ref{alg:input_selection} and $Q(\emptyset) = \lambda_{\min}(R)$ from Eqn.~\eqref{eq:Sub_Q}.
	\end{proof}
\end{prop}

%\begin{rem}
%	Selecting input nodes to achieve phase synchronization in \eqref{eq:Kuramoto with Inputs} requires input of Algorithm~\ref{alg:input_selection} to be a non-symmetric matrix $\bar{M}$. Hence, it is required to chose an arbitrary diagonal matrix $\bar{D}$ and set the input of Algorithm~\ref{alg:input_selection} to be $\bar{M}^{T}\bar{D}+\bar{D}\bar{M}$ to obtain the minimum set of input nodes (Lemma 3 in \cite{clark2018maximizing}).
%\end{rem}
\section{Numerical Study}\label{sec:sim} 
In this section we demonstrate the performance of Algorithm~\ref{alg:input_selection} in selecting a minimum set of control inputs for synchronizing Kuramoto networks. %In the following we organize the simulation results of Algorithm~\ref{alg:input_selection} in to two groups: Homogeneous natural frequency case (Section~\ref{subsec:sim1}) and Heterogeneous natural frequency case (Section~\ref{subsec:sim2}). 

In what follows, we refer to Algorithm~\ref{alg:input_selection} by \textit{submodular} algorithm for the purpose of comparing the performance of Algorithm~\ref{alg:input_selection} against two other selection algorithms: \textit{greedy} and \textit{random}.  In each iteration greedy algorithm selects an oscillator $i \in \{1, \ldots, n\}$ such that removing rows and columns corresponding to the oscillator $i$ maximizes the $\lambda_{\min}(R)$ until $\lambda_{\min}(R(S)) > \delta$. The random algorithm selects an oscillator uniformly at random in each iteration until $\lambda_{min}(R(S)) > \delta$. We compare the submodular, random, and greedy algorithms with the exact optimal set of inputs, which is computed via exhaustive search.

%We also compare the performance of submodular algorithm against the ground truth values returned by exhaustively searching for the minimum number of inputs required for the synchronization among  $2^{n}-1$ possible combinations of oscillators in the network. For simplicity, we call this exhaustive search algorithm as \textit{optimal} algorithm.

The simulations were performed for three types of 10 node graphs, namely, Undirected, directed oriented and directed oriented cycle graphs. Each of the data point in the plots of Figure~\ref{fig:sim1} and Figure~\ref{fig:sim2} corresponds to 100 random realizations of each graph type considered.

%We fix the number of nodes in $\mathcal{G}$ to be 10.
%\subsection{Homogeneous Kuramoto networks}\label{subsec:sim1}

%In here we vary the fraction of negative edges in $\mathcal{G}$ and compare the number of control input nodes returned by submodular algorithm against greedy, random and optimal algorithms in Figure~\ref{fig:sim1}. We keep the connectivity and coupling weights of each $\mathcal{G}$ around similar values for each graph type. 

In Figure~\ref{fig:sim1}, we choose coupling weights of each $\mathcal{G}$ randomly from the interval $[1, 5]$ and then set the number of negative edges present in $\mathcal{G}$ according to the fraction of negative edges considered. The results  suggest that submodular algorithm outperforms the random algorithm in all three graph types compared. The number of control inputs selected by submodular and greedy algorithm are comparable in undirected and directed oriented graph types.  However in directed cyclic graphs we observe that submodular algorithm outperforms both greedy and random algorithms. The average difference between the number of control inputs chosen by subbmodular and optimal algorithms is $0.76$.

\begin{figure*}
	\centering
	%\begin{subfigure}[t]{0.245\textwidth}
	%	\raisebox{-\height}{\includegraphics[width=\textwidth]{Figures/Random_NegPec.jpg}}
	%	\caption{}
	%\end{subfigure}
	%\hfill
	\begin{subfigure}[t]{0.32\textwidth}
		\raisebox{-\height}{\includegraphics[width=\textwidth]{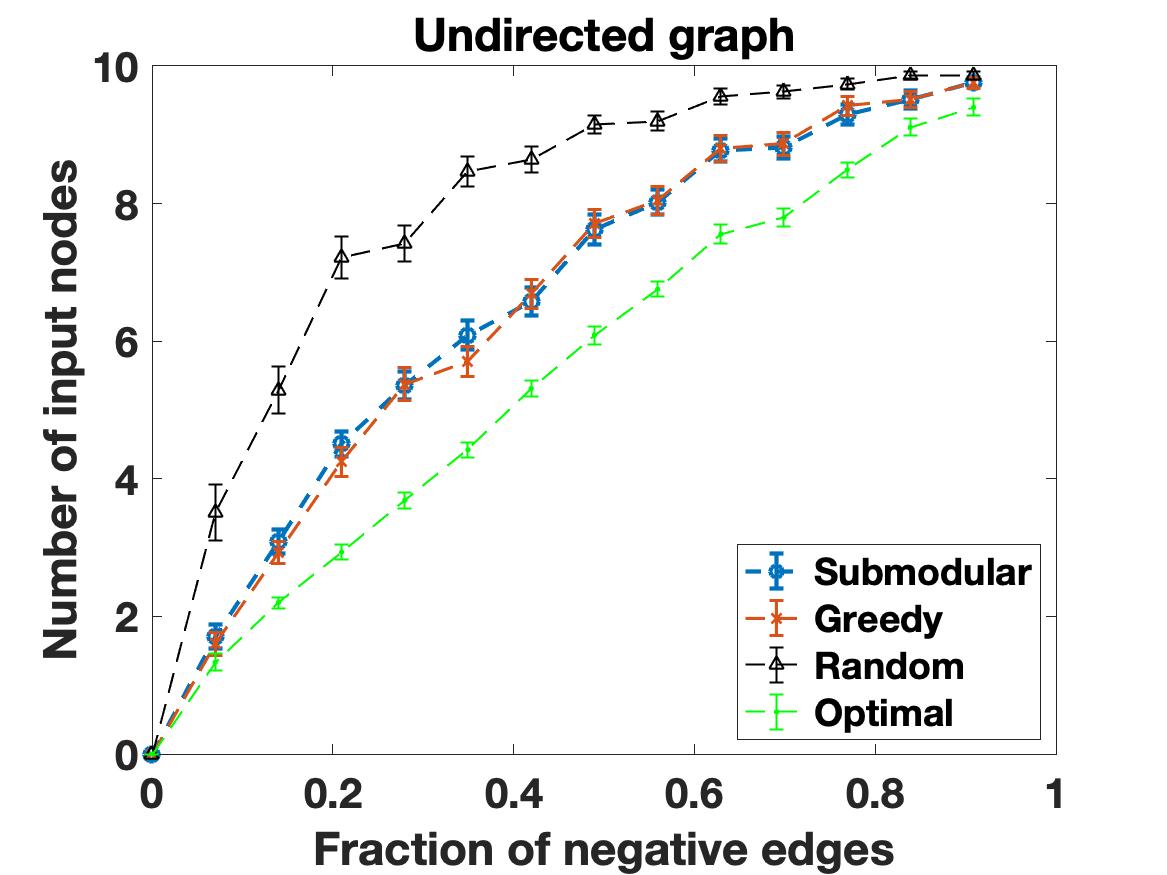}}
		\caption{}
	\end{subfigure}
    %\hfill
	\begin{subfigure}[t]{0.32\textwidth}
		\raisebox{-\height}{\includegraphics[width=\textwidth]{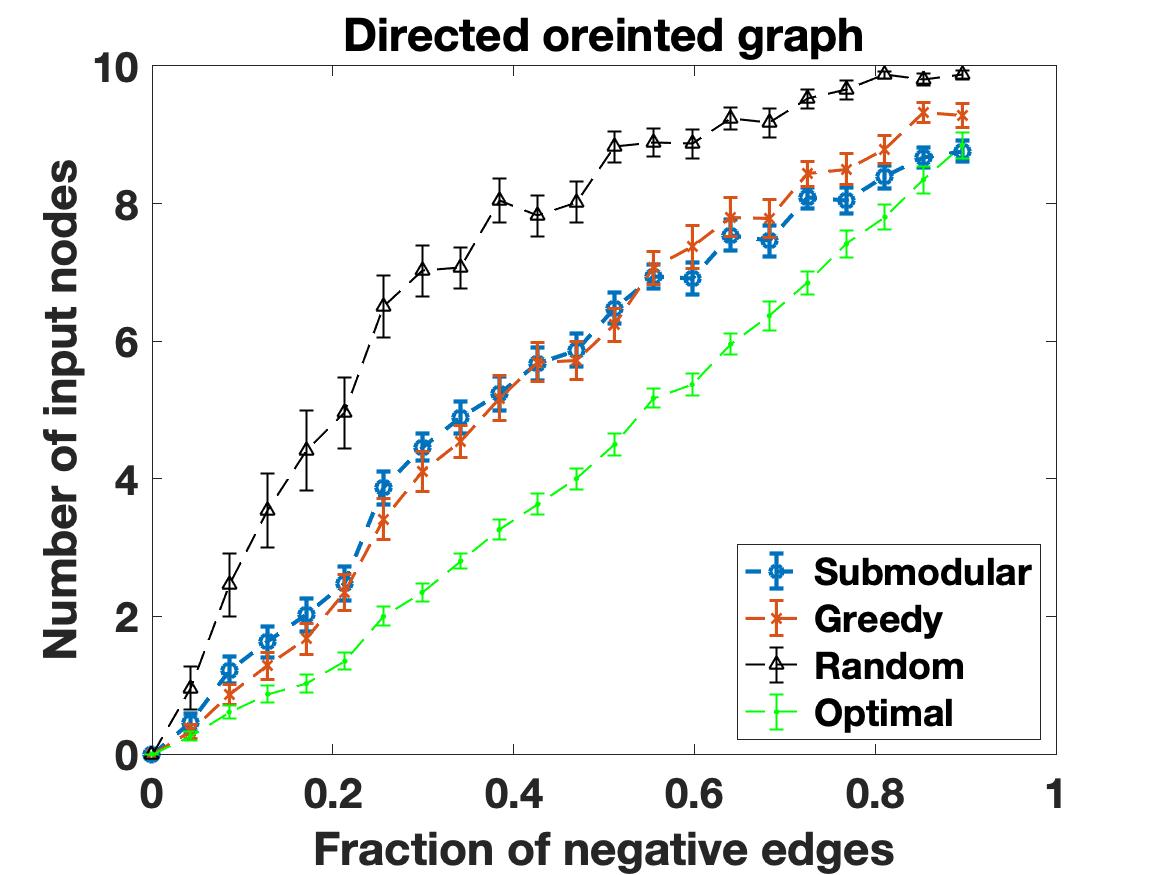}}
		\caption{}
	\end{subfigure}
	%\hfill
	\begin{subfigure}[t]{0.32\textwidth}
		\raisebox{-\height}{\includegraphics[width=\textwidth]{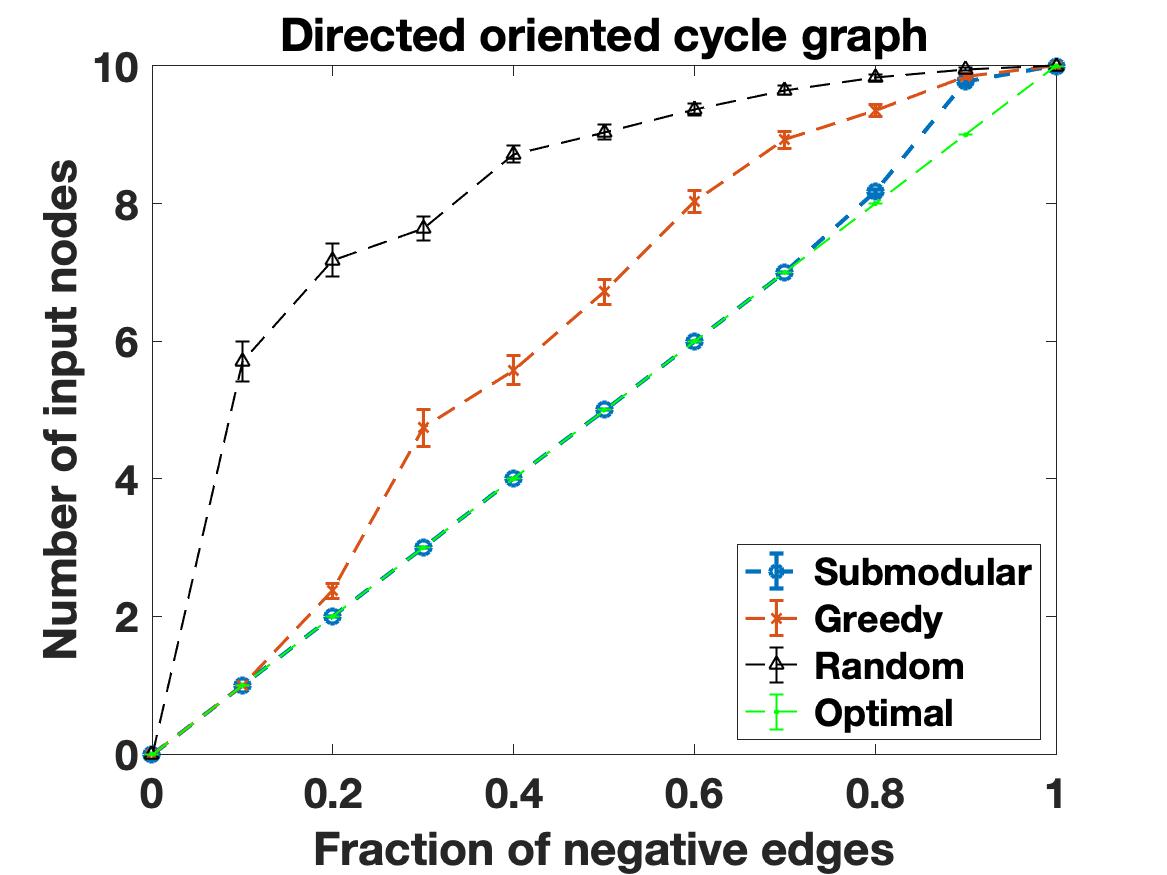}}
		\caption{}
	\end{subfigure}
\caption{For homogeneous Kuramoto networks, number of control input nodes as a function of  negative edges for 10 node graphs. 
	We compare performance of submodular, greedy, random and optimal algorithms with each data point averaged over 100 realizations.}
%Figure~(a) Undirected, Figure~(b) directed oriented, Figure~(c)  directed oriented cycle graphs. 
	%Minimum-size set of control input selection in homogeneous Kuramoto networks: Under each graph type we consider graphs of 10 nodes and vary the fraction of negative edges while keeping the the connectivity and coupling weights of each $\mathcal{G}$ around similar values. Then of.}%: Comparison of the number of input nodes selected  by submodular, greedy and random selection algorithms to frequency/phase synchronize Kuramoto network with identical natural frequencies under different fractions of negative edges in each graph type. Sub figures~(a),~(b),~(c)~and~(d) illustrate the cases where underlying Kuramoto network graph $\mathcal{G}$ is random, undirected, directed oriented and directed oriented cycle graph, respectively. In sub figure (d) plot related to submodular input selection overlaps with the  plot showing optimal inputs.}
\label{fig:sim1}
\end{figure*}

%\subsection{Heterogeneous networks}\label{subsec:sim2}
%used for the simulations 
In the simulations related to the heterogeneous Kuramoto networks, we use the parameter called Weight-Frequency (WF) parameter which captures the ratio between inter oscillator natural frequency differences $(D^{T}\omega)$ and the in-degree ($d_{\text{in}}(i) = \sum_{(j,i) \in E(i)} K_{ji}$) of an oscillator $i$. WF parameter is defined as 
\begin{equation}\label{eq:WF}
\text{WF} =  ||D^{T}\omega||_{2}/\max(d_{\text{in}}(i)).
\end{equation}

% for $i \in \{1, \ldots, n\}$.
%Then we scale the coupling weights to vary the WF parameter in each $\mathcal{G}$ and compare the number of control input nodes returned by submodular algorithm against greedy, random and optimal algorithms in Figure~\ref{fig:sim2}. We keep the fraction of negative edges, connectivity and natural frequencies of each $\mathcal{G}$ around similar values for each graph type. In particular, 

In Figure~\ref{fig:sim2} to we consider minimum-set control input selection for synchronization in heterogeneous Kuramoto networks. In here we set fraction of negative edges to be $0.3$ and randomly selects the natural frequencies of  oscillators from the interval $[0, 2]$. The results suggests that the performance of submodular algorithm is very similar to the results observed in homogeneous Kuramoto networks. The average difference between the number of control inputs chosen by submodular and optimal algorithm is $1.39$. 

%Number of control inputs required for synchronization is increased compared to the case of homogeneous natural frequencies. This is due to the fact that the conditions require for synchronization in heterogeneous case is stricter than homogeneous case ($\delta = 0$ vs. $\delta = D^{T}\omega$). 
%(higher than in the cases analyzed in Section~\ref{subsec:sim1})

\begin{figure*}
	\vspace{-1mm}
	\centering
	%\begin{subfigure}[t]{0.245\textwidth}
	%	\raisebox{-\height}{\includegraphics[width=\textwidth]{Figures/Random_WF.jpg}}
	%	\caption{}
%	\end{subfigure}
	%\hfill
	\begin{subfigure}[t]{0.32\textwidth}
		\raisebox{-\height}{\includegraphics[width=\textwidth]{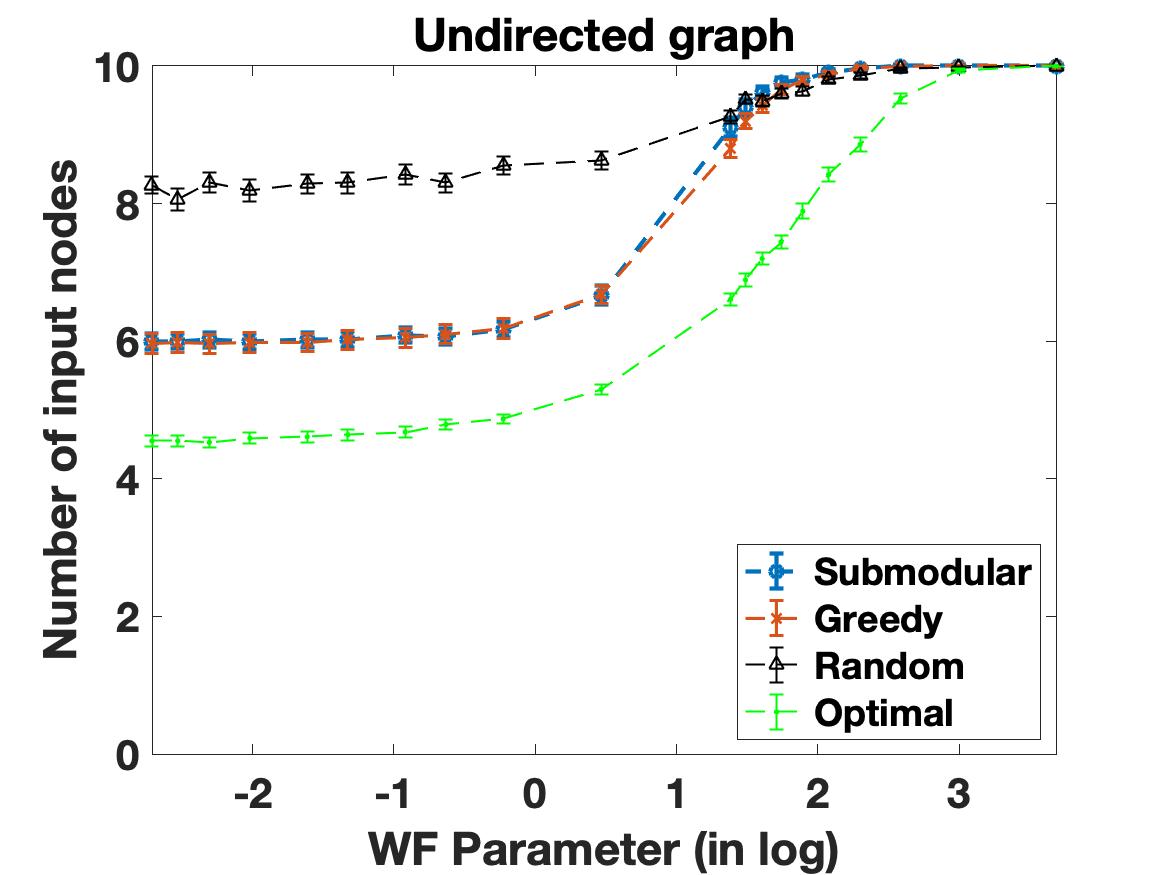}}
		\caption{}
	\end{subfigure}
	%\hfill
	\begin{subfigure}[t]{0.32\textwidth}
		\raisebox{-\height}{\includegraphics[width=\textwidth]{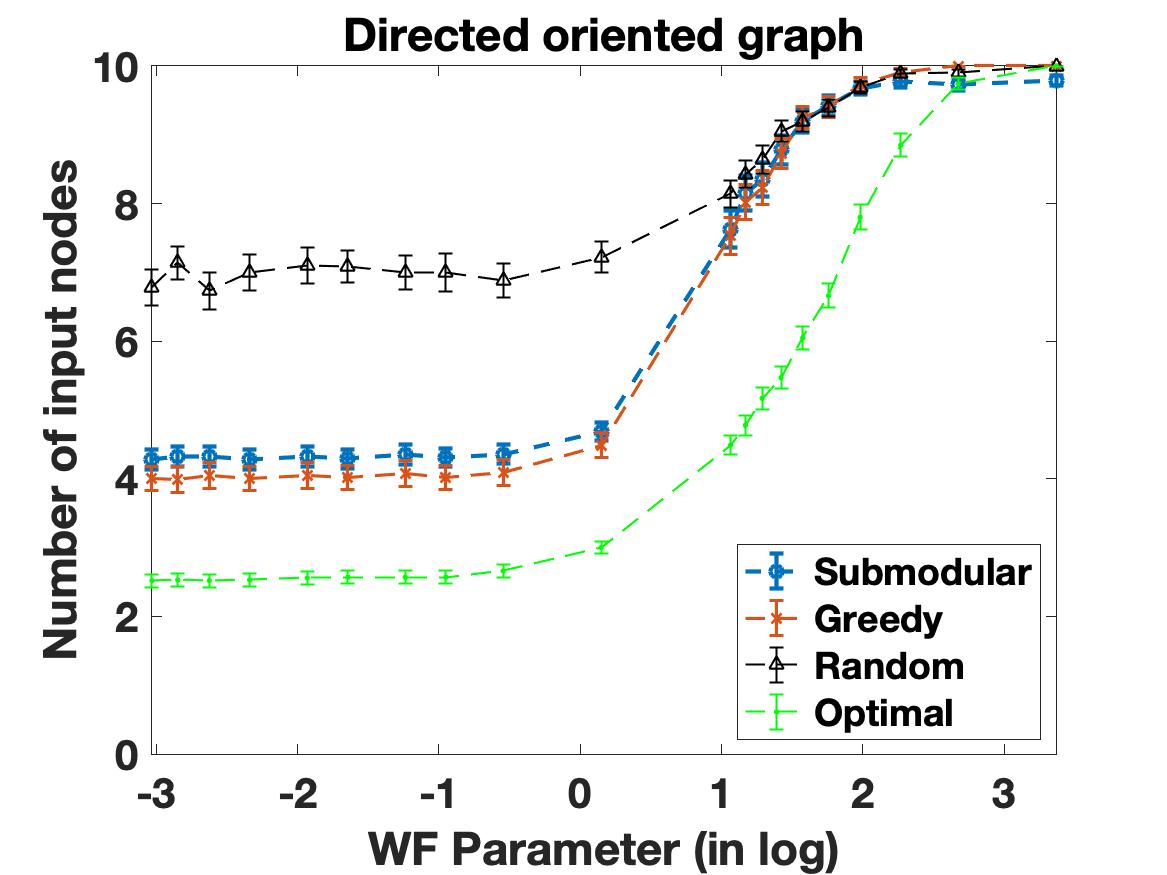}}
		\caption{}
	\end{subfigure}
	%\hfill
	\begin{subfigure}[t]{0.32\textwidth}
		\raisebox{-\height}{\includegraphics[width=\textwidth]{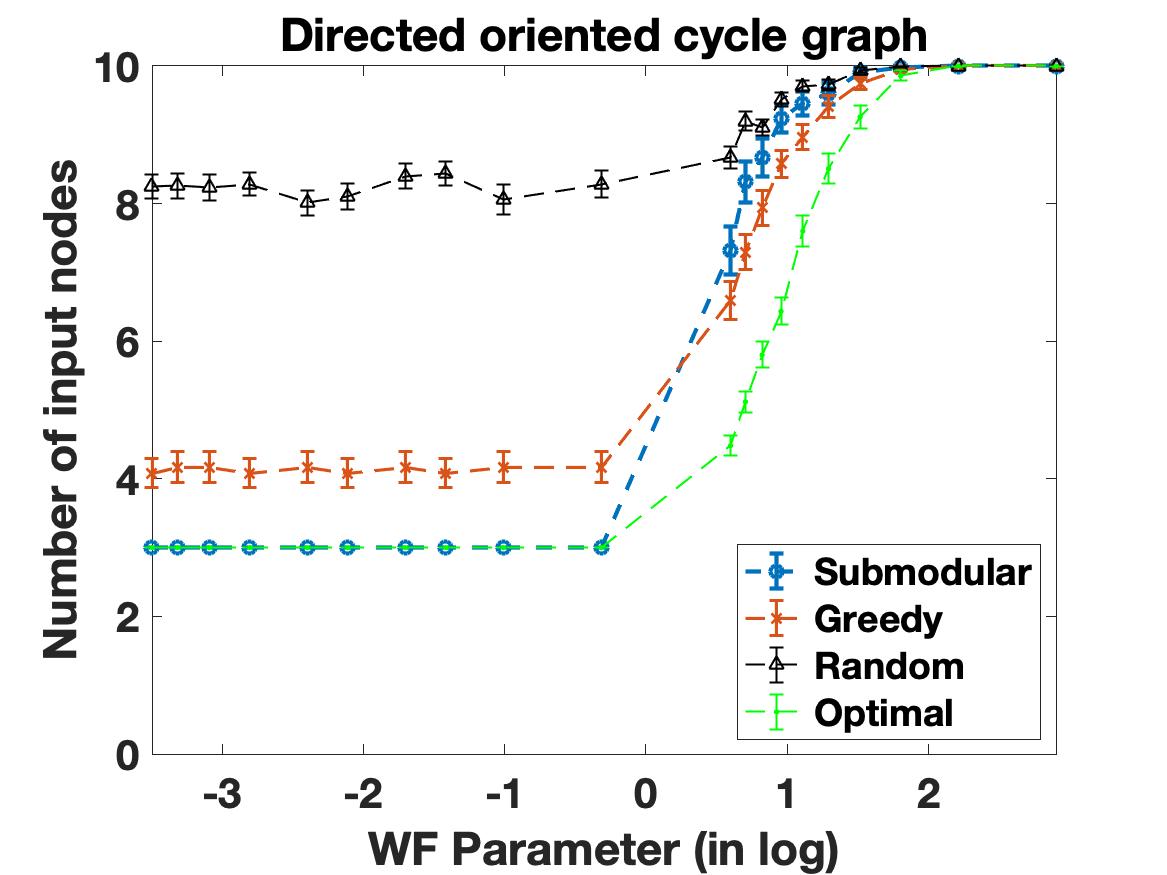}}
		\caption{}
	\end{subfigure}
	\caption{For heterogeneous Kuramoto networks, number of control input nodes as function of  WF parameter defined in Eqn.~\ref{eq:WF} for 10 node graphs. We compare performance of submodular, greedy, random and optimal algorithms with each data point averaged over 100 realizations.}
		
		%For graphs~(a) Undirected, (b) directed oriented and (c) directed oriented cycle graphs.
		%Minimum-size set of control input selection in heterogeneous Kuramoto networks: Under each graph type we consider graphs of 10 nodes and vary the WF parameter defined in Eqn.~\ref{eq:WF} while keeping the the fraction of negative edges, connectivity and natural frequencies of each $\mathcal{G}$ around similar values. Then we compare the number of control inputs selected by submodular, greedy, random and optimal algorithms. Each data point is averaged over 100 realizations of graphs.}% Comparison of the number of input nodes selected  by submodular, greedy and random selection algorithms to frequency/phase synchronize Kuramoto network with identical natural frequencies under different fractions of negative edges in each graph type. Sub figures~(a),~(b),~(c)~and~(d) illustrate the cases where underlying Kuramoto network graph $\mathcal{G}$ is random, undirected, directed oriented and directed oriented cycle graph, respectively. In sub figure (d) plot related to submodular input selection overlaps with the  plot showing optimal inputs.}
	\label{fig:sim2}
	\vspace{-5mm}
\end{figure*}

\section{Conclusions}\label{sec:end}
We studied the problem of minimum-set control input selection for synchronization of Kuramoto networks. We developed a passivity-based analytical framework to obtain sufficient conditions for synchronization. Our framework enables control input selection in homogeneous and heterogeneous Kuramoto networks with signed couplings under various network typologies. We developed a submodular optimization algorithm for selecting a minimum set of control inputs needed to drive Kuramoto networks towards synchronization. We simulated our algorithm on undirected, directed oriented, and directed oriented cycle graphs.  Future work includes the design of time-varying input signals associated with a minimum set of control inputs in heterogeneous Kuramoto networks for achieving synchronization. 

%Furthermore, in the case of heterogeneous natural frequencies, we provided conditions on initial phase angles such that phase angle differences related to positive and negative couplings at frequency synchronization are trapped within the regions $[-\pi/2, \pi/2]$ and $[\pi/2, 3\pi/2]$, respectively.

%%%%%%%%% REFERENCES
%{\small
\bibliographystyle{IEEEtran}
\bibliography{CDC2020}

\begin{thebibliography}{10}
\providecommand{\url}[1]{#1}
\csname url@samestyle\endcsname
\providecommand{\newblock}{\relax}
\providecommand{\bibinfo}[2]{#2}
\providecommand{\BIBentrySTDinterwordspacing}{\spaceskip=0pt\relax}
\providecommand{\BIBentryALTinterwordstretchfactor}{4}
\providecommand{\BIBentryALTinterwordspacing}{\spaceskip=\fontdimen2\font plus
\BIBentryALTinterwordstretchfactor\fontdimen3\font minus
  \fontdimen4\font\relax}
\providecommand{\BIBforeignlanguage}[2]{{%
\expandafter\ifx\csname l@#1\endcsname\relax
\typeout{** WARNING: IEEEtran.bst: No hyphenation pattern has been}%
\typeout{** loaded for the language `#1'. Using the pattern for}%
\typeout{** the default language instead.}%
\else
\language=\csname l@#1\endcsname
\fi
#2}}
\providecommand{\BIBdecl}{\relax}
\BIBdecl

\bibitem{dorfler2013synchronization}
F.~D{\"o}rfler, M.~Chertkov, and F.~Bullo, ``Synchronization in complex
  oscillator networks and smart grids,'' \emph{Proceedings of the National
  Academy of Sciences}, vol. 110, no.~6, pp. 2005--2010, 2013.

\bibitem{schenato2007distributed}
L.~Schenato and G.~Gamba, ``A distributed consensus protocol for clock
  synchronization in wireless sensor network,'' in \emph{46th IEEE Conference
  on Decision and Control}.\hskip 1em plus 0.5em minus 0.4em\relax IEEE, 2007,
  pp. 2289--2294.

\bibitem{bennett2002huygens}
M.~Bennett, M.~F. Schatz, H.~Rockwood, and K.~Wiesenfeld, ``Huygens's clocks,''
  \emph{Proceedings of the Royal Society of London. Series A: Mathematical,
  Physical and Engineering Sciences}, vol. 458, no. 2019, pp. 563--579, 2002.

\bibitem{herzog2007neurons}
E.~D. Herzog, ``Neurons and networks in daily rhythms,'' \emph{Nature Reviews
  Neuroscience}, vol.~8, no.~10, pp. 790--802, 2007.

\bibitem{torre1976theory}
V.~Torre, ``A theory of synchronization of heart pace-maker cells,''
  \emph{Journal of theoretical biology}, vol.~61, no.~1, pp. 55--71, 1976.

\bibitem{kuramoto1975self}
Y.~Kuramoto, ``Self-entrainment of a population of coupled non-linear
  oscillators,'' in \emph{International symposium on mathematical problems in
  theoretical physics}.\hskip 1em plus 0.5em minus 0.4em\relax Springer, 1975,
  pp. 420--422.

\bibitem{acebron2005kuramoto}
J.~A. Acebr{\'o}n, L.~L. Bonilla, C.~J.~P. Vicente, F.~Ritort, and R.~Spigler,
  ``The {K}uramoto model: A simple paradigm for synchronization phenomena,''
  \emph{Reviews of modern physics}, vol.~77, no.~1, p. 137, 2005.

\bibitem{jadbabaie2004stability}
A.~Jadbabaie, N.~Motee, and M.~Barahona, ``On the stability of the {K}uramoto
  model of coupled nonlinear oscillators,'' in \emph{Proceedings of the
  American Control Conference}, vol.~5.\hskip 1em plus 0.5em minus 0.4em\relax
  IEEE, 2004, pp. 4296--4301.

\bibitem{strogatz2000kuramoto}
S.~H. Strogatz, ``From {K}uramoto to {C}rawford: exploring the onset of
  synchronization in populations of coupled oscillators,'' \emph{Physica D:
  Nonlinear Phenomena}, vol. 143, no. 1-4, pp. 1--20, 2000.

\bibitem{clark2017toward}
A.~Clark, B.~Alomair, L.~Bushnell, and R.~Poovendran, ``Toward synchronization
  in networks with nonlinear dynamics: A submodular optimization framework,''
  \emph{IEEE Transactions on Automatic Control}, vol.~62, no.~10, pp.
  5055--5068, 2017.

\bibitem{bronski2012fully}
J.~C. Bronski, L.~DeVille, and M.~Jip~Park, ``Fully synchronous solutions and
  the synchronization phase transition for the finite-n {K}uramoto model,''
  \emph{Chaos: An Interdisciplinary Journal of Nonlinear Science}, vol.~22,
  no.~3, p. 033133, 2012.

\bibitem{rogge2004stability}
J.~Rogge and D.~Aeyels, ``Stability of phase locking in a ring of
  unidirectionally coupled oscillators,'' \emph{Journal of Physics A:
  Mathematical and General}, vol.~37, no.~46, p. 11135, 2004.

\bibitem{mesbahi2010graph}
M.~Mesbahi and M.~Egerstedt, \emph{{G}raph Theoretic Methods in Multiagent
  Networks}.\hskip 1em plus 0.5em minus 0.4em\relax Princeton University Press,
  2010.

\bibitem{hong2011kuramoto}
H.~Hong and S.~H. Strogatz, ``Kuramoto model of coupled oscillators with
  positive and negative coupling parameters: an example of conformist and
  contrarian oscillators,'' \emph{Physical Review Letters}, vol. 106, no.~5, p.
  054102, 2011.

\bibitem{el2013synchronization}
A.~El~Ati and E.~Panteley, ``Synchronization of phase oscillators with
  attractive and repulsive interconnections,'' in \emph{2013 18th International
  Conference on Methods \& Models in Automation \& Robotics (MMAR)}.\hskip 1em
  plus 0.5em minus 0.4em\relax IEEE, 2013, pp. 22--27.

\bibitem{delabays2019kuramoto}
R.~Delabays, P.~Jacquod, and F.~Dorfler, ``The {K}uramoto model on oriented and
  signed graphs,'' \emph{SIAM Journal on Applied Dynamical Systems}, vol.~18,
  no.~1, pp. 458--480, 2019.

\bibitem{li2015synchronizing}
X.~Li and P.~Rao, ``Synchronizing a weighted and weakly-connected
  {K}uramoto-oscillator digraph with a pacemaker,'' \emph{IEEE Transactions on
  Circuits and Systems I: Regular Papers}, vol.~62, no.~3, pp. 899--905, 2015.

\bibitem{wang2012exponential}
Y.~Wang and F.~J. Doyle, ``Exponential synchronization rate of {K}uramoto
  oscillators in the presence of a pacemaker,'' \emph{IEEE transactions on
  automatic control}, vol.~58, no.~4, pp. 989--994, 2012.

\bibitem{bosso2019global}
A.~Bosso, I.~A. Azzollini, and S.~Baldi, ``Global frequency synchronization
  over networks of uncertain second-order {K}uramoto oscillators via
  distributed adaptive tracking,'' in \emph{IEEE 58th Conference on Decision
  and Control (CDC)}.\hskip 1em plus 0.5em minus 0.4em\relax IEEE, 2019, pp.
  1031--1036.

\bibitem{brogliato2007dissipative}
B.~Brogliato, R.~Lozano, B.~Maschke, and O.~Egeland, \emph{Dissipative Systems
  Analysis and Control}.\hskip 1em plus 0.5em minus 0.4em\relax Springer, 2007,
  vol.~2.

\bibitem{moreau2004stability}
L.~Moreau, ``Stability of continuous-time distributed consensus algorithms,''
  in \emph{43rd IEEE conference on decision and control (CDC)}, vol.~4.\hskip
  1em plus 0.5em minus 0.4em\relax IEEE, 2004, pp. 3998--4003.

\bibitem{schmidt2012frequency}
G.~S. Schmidt, A.~Papachristodoulou, U.~M{\"u}nz, and F.~Allg{\"o}wer,
  ``Frequency synchronization and phase agreement in {K}uramoto oscillator
  networks with delays,'' \emph{Automatica}, vol.~48, no.~12, pp. 3008--3017,
  2012.

\bibitem{clark2018maximizing}
A.~Clark, Q.~Hou, L.~Bushnell, and R.~Poovendran, ``Maximizing the smallest
  eigenvalue of a symmetric matrix: A submodular optimization approach,''
  \emph{Automatica}, vol.~95, pp. 446--454, 2018.

\bibitem{wolsey1982analysis}
L.~A. Wolsey, ``An analysis of the greedy algorithm for the submodular set
  covering problem,'' \emph{Combinatorica}, vol.~2, no.~4, pp. 385--393, 1982.

\end{thebibliography}
%\bibliography{DIFT-Reference}
%}
%\nocite{*}

\appendix\label{appendix}
In this section, we provide proofs of Theorem~\ref{thm:Undirected_init}, Corollary~\ref{cor:Cycle init} and  Corollary~\ref{cor:Tree_init}.

\textit{Proof of Theorem~\ref{thm:Undirected_init}: }
%\begin{proof*}
In this proof we let $\mathcal{G}$ to be undirected.
%Consider two nodes $i$ and $j$ in undirected $\mathcal{G}$.
Notice that solution to the set of inequalities given in  Assumption~\ref{asmp:Phase_Angle_Init} does exists if all the distinct inequalities involving any two nodes $i$ and $j$ in $\mathcal{G}$ such that $(i,j) \in E$ do not contradict with each other. Also note that the number of distinct inequalities involving any two nodes $i$ and $j$ in $\mathcal{G}$ is equal to the number of distinct paths between node $i$ and $j$ in $\mathcal{G}$.

First consider the condition~$1)$ where $K_{ij} > 0$ with $z_{ij}(0) = \theta_{i}(0) - \theta_{j}(0)$. Then from Assumption~\ref{asmp:Phase_Angle_Init} we have 
\begin{equation}\label{eq:undirected_01}
-\pi/2 < z_{ij}(0) < \pi/2. 
\end{equation}
Then by adding all the inequalities defined according to Assumption~\ref{asmp:Phase_Angle_Init} corresponding to the edges in the path $l$ between node $i$ and node $j$ gives
\begin{equation}\label{eq:undirected_02}
\begin{split}
-E_{(p-n)}\pi/2 < z_{ij} &< E_{(p-n)}\pi/2~\text{if}~E_{(p-n)}\geq 0 \\ 
E_{(n-p)}\pi/2 < z_{ij} &< E_{(n-p)}3\pi/2 \text{if}~E_{(n-p)} \geq 0,
\end{split}
\end{equation}
where $E_{(p-n)} = (E_p^{l}(i,j) - E_n^{l}(i,j))\mod4$ and $E_{(n-p)} = (E_n^{l}(i,j) - E_p^{l}(i,j))\mod4$. Hence in this case we require $E_{(p-n)} \in \{0,1\}$ or $E_{(n-p)} \in \{1,3\}$ in order for the solutions of the inequalities in Eqns.~\eqref{eq:undirected_01}~and~\eqref{eq:undirected_02}  to coincide.

Next consider the condition~$2)$ where $K_{ij} < 0$. Then from Assumption~\ref{asmp:Phase_Angle_Init} we have 
\begin{equation}\label{eq:undirected_03}
\pi/2 < z_{ij}(0) < 3\pi/2. 
\end{equation}
Then similar to the proof of condition one we can obtain the set of inequalities in Eqn.~\eqref{eq:undirected_02} by summing up all the inequalities corresponding to the edges in the path $l$ between node $i$ and node $j$. Therefore, we require $E_{(p-n)} \in \{2,3\}$ or $E_{(n-p)} \in \{1,2\}$ in order for the solutions of the inequalities in Eqns.~\eqref{eq:undirected_02}~and~\eqref{eq:undirected_03}  to coincide.

Similar arguments follows when the Kuramoto network's graph structure is a directed oriented graph. $\null\hfill{\blacksquare}$
%\end{proof*}

\textit{Proof of Corollary~\ref{cor:Cycle init}:}
%\begin{proof}
Notice that in cycle graphs for each pair of nodes $(i,i+1)$ for $i \in {1, \ldots, n-1}$ and pair $(n,1)$ has exactly one path whose length is greater than one. Then the proof follows form the similar arguments as in the proof of Theorem~\ref{thm:Undirected_init} with $l=1$. $\null\hfill{\blacksquare}$
%\end{proof}

\textit{Proof of Corollary~\ref{cor:Tree_init}: }
%\begin{proof}
Tree graphs has exactly one path between any pair of nodes $(i,j)$. Hence if $(i,j) \in E$ then there exists no other path between node $i$ and node $j$ such that path length is greater than one. $\null\hfill{\blacksquare}$
%\end{proof}

\end{document}